\documentclass[journal]{IEEEtran}

\usepackage{amsmath, mathtools, mathrsfs, dsfont, cases}

\usepackage{algorithm}
\usepackage{algorithmic}

\usepackage{adjustbox}

\usepackage{amsthm}
\usepackage{amssymb}
\usepackage{cite, url}

\usepackage{booktabs, tabu, array, makecell}
\usepackage{multirow, multicol}
\usepackage{setspace}
\usepackage[table, dvipsnames]{xcolor}  % This includes color functionality
\newcolumntype{?}{!{\vrule width 1.2pt}}

\usepackage{graphicx, epstopdf, epsfig}
\usepackage{float}
\usepackage{subcaption}
\usepackage[caption = false]{subfig}
\usepackage[font=small, labelfont=bf, justification=centering]{caption}
\usepackage{stfloats}
\usepackage{tikz}
\usetikzlibrary{arrows.meta, positioning, shapes.geometric, fit, calc, shapes.multipart, tikzmark}

\usepackage{adjustbox, stackengine, verbatim, enumitem, boldline, bm}
\usepackage[utf8]{inputenc}
\usepackage{adjustbox, stackengine, verbatim, enumitem, boldline, bm}
\usepackage[utf8]{inputenc}
\renewcommand\qedsymbol{$\blacksquare$}

\newtheorem{remark}{Remark}
\newtheorem{theorem}{Theorem}[section]

\newtheorem{corollary}[theorem]{Corollary}

\usetikzlibrary{arrows.meta, positioning}

\usepackage{placeins}
\usetikzlibrary{tikzmark}

\usepackage[english]{babel}

\tikzset{mycircled/.style={circle,draw,inner sep=0.1em,line width=0.1em}}

\definecolor{bb}{rgb}{0.2941, 0.5447, 0.7494}
\definecolor{orange}{rgb}{1, 0.347, 0}
\definecolor{greenJ}{rgb}{0, 0.6590, 0.42}
\definecolor{zereshk}{rgb}{0.588,0,0.098}
\definecolor{Maroon}{rgb}{0.502, 0, 0}
\definecolor{Brown}{rgb}{0.588, 0.294, 0}
\definecolor{Olive}{rgb}{0.502, 0.502, 0}
\definecolor{Navy}{rgb}{0, 0, 0.502}
\definecolor{Orange}{rgb}{1,0.647,0}
\definecolor{Yellow}{rgb}{0.502, 1, 0}
\definecolor{Green}{rgb}{0, 0.502,  0}
\definecolor{Blue}{rgb}{0, 0,0.761}
\definecolor{Lime}{rgb}{0.196, 0.804, 0.196}
\definecolor{Purple}{rgb}{0.502,0,0.502}
\definecolor{Violet}{rgb}{0.561,0,1}
\definecolor{Magneta}{rgb}{1,0,1}
\definecolor{Red}{rgb}{1,0,0}
\definecolor{Abi}{rgb}{0.059, 0.322, 0.729}
\definecolor{Gray}{rgb}{0.502, 0.502, 0.502}
\definecolor{um}{rgb}{0.0824, 0.1294, 0.4196}
\definecolor{abikam}{rgb}{0.51, 0.93,  0.992}
\definecolor{abizeyad}{rgb}{0.16, 0.573,0.761}
\definecolor{rr}{rgb}{0.9047, 0.1918, 0.1988}

\definecolor{orange2}{rgb}{0.85,0.33,0.10}
\definecolor{yellow4}{rgb}{0.93,0.69,0.13}
\definecolor{purple6}{rgb}{0.49,0.18,0.56}
\definecolor{green32}{rgb}{0.47,0.67,0.19}
\definecolor{blue13}{rgb}{0.30,0.75,0.93}
\definecolor{zereshk16}{rgb}{0.64,0.08,0.18}
\definecolor{orange25}{rgb}{0.93,0.69,0.13}
\definecolor{gray27}{rgb}{0.0824, 0.1294, 0.4196}

\definecolor{storageBlue}{rgb}{0, 0.4470, 0.7410}   % For Storage ID 4
\definecolor{storageOrange}{rgb}{0.8500, 0.3250, 0.0980} % For Storage ID 10
\definecolor{storageYellow}{rgb}{0.9290, 0.6940, 0.1250} % For Storage ID 18
\definecolor{storagePurple}{rgb}{0.4940, 0.1840, 0.5560} % For Storage ID 27
\definecolor{storageGreen}{rgb}{0.4660, 0.6740, 0.1880} % For Storage ID 33

\title{Real-Time Defense Against Coordinated Cyber-Physical Attacks: A Robust Constrained Reinforcement Learning Approach}
\IEEEaftertitletext{\vspace{-2.6\baselineskip}}
\allowdisplaybreaks
\author{Saman~Mazaheri~Khamaneh,~\IEEEmembership{Student~Member,~IEEE,}
        Tong~Wu,~\IEEEmembership{Member,~IEEE,} ~Wei~Sun,~\IEEEmembership{Senior Member,~IEEE}, and ~Cong Chen,~\IEEEmembership{Member,~IEEE} % <-this % stops a space
}

\newcommand{\abs}[1]{\lvert #1 \rvert}
\begin{document}
\maketitle

\begin{abstract}
Modern power systems face increasing vulnerability to sophisticated cyber-physical attacks beyond traditional \textit{N-1} contingency frameworks. Existing security paradigms face a critical bottleneck: efficiently identifying worst-case scenarios and rapidly coordinating defensive responses are hindered by intensive computation and time delays, during which cascading failures can propagate.
This paper presents a novel tri-level robust constrained reinforcement learning (RCRL) framework for robust power system security. The framework generates diverse system states through AC-OPF formulations, identifies worst-case \textit{N-K} attack scenarios for each state, and trains policies to mitigate these scenarios across all operating conditions without requiring predefined attack patterns. The framework addresses constraint satisfaction through Beta-blending projection-based feasible action mapping techniques during training and primal-dual augmented Lagrangian optimization for deployment. Once trained, the RCRL policy learns how to control observed cyber-physical attacks in real time.
Validation on IEEE benchmark systems demonstrates effectiveness against coordinated \textit{N-K} attacks, causing widespread cascading failures throughout the network. The learned policy can successfully respond rapidly to recover system-wide constraints back to normal ranges within 0.21 ms inference times, establishing superior resilience for critical infrastructure protection.
\end{abstract}

\vspace{-3mm}

\section{Introduction} \label{sec:introduction}
\vspace{-1mm}
\subsection{Background}
Modern power system networks are increasingly vulnerable to sophisticated cyber-physical attacks due to their growing dependence on digital communication and control systems. Advanced grid technologies such as wide-area monitoring systems   and automated protection schemes, while improving operational efficiency, have significantly expanded the attack surface available to malicious actors. A prime example is the 2015-2016 Ukraine cyberattacks, which demonstrated how coordinated targeting of transmission substations and generation units can trigger cascading failures that impact millions of people~\cite{sullivan2017}. These incidents reveal fundamental inadequacies in conventional security paradigms. Existing approaches rely on independent failure assumptions and cannot address coordinated, intelligent adversarial behavior targeting transmission assets~\cite{khamaneh2025robust, wang2013}. Consequently, addressing these emerging threats requires adaptive defense mechanisms capable of real-time response to worst-case coordinated attack scenarios.

Current industry practices rely on \textit{N-1} contingency analysis, which assumes random and independent component failures~\cite{majidi2015integration}. This approach is fundamentally inadequate for coordinated cyber-physical attacks where disruptions are neither independent nor random~\cite{du2019achieving}. Furthermore, \textit{N-1}-based methods face three critical limitations: (i) computational intensity prevents real-time implementation under adversarial conditions~\cite{donti2021adversarially, mohammadi2016agent}, (ii) reactive control strategies suffer from optimization delays that result in control actions being applied to evolved system states rather than original conditions, extending outage durations~\cite{wu2023constrained}, and (iii) traditional strategies rely on generator inertia, which is diminishing in modern transmission systems with increased renewable penetration~\cite{gordon2021impact}. These limitations necessitate proactive, learning-based frameworks capable of real-time adaptation to evolving threats. 

\vspace{-0.3cm}
\subsection{Related Work}
To address these challenges, existing approaches to cyber-physical attacks, including classical optimization methods and learning-based algorithms, can be organized into the following:
\subsubsection{Power System Security Against Cyber-Physical Attacks}
Power system cybersecurity research has evolved from foundational threat identification~\cite{sridhar2011cyber} to sophisticated defense mechanisms against coordinated adversarial attacks. Current approaches fall into two categories: reactive cyber-attack mitigation and proactive security-constrained optimization. Reactive mitigation strategies include wide-area monitoring systems~\cite{ashok2017cyber}, game-theoretic attack-defense analysis~\cite{lakshminarayana2019performance,lakshminarayana2017optimal}, and coordinated risk mitigation frameworks~\cite{zhang2021cyber}. Proactive approaches focus on \textit{N-K} security-constrained optimal power flow (SCOPF) formulations~\cite{lai2019tri}, large-scale bilevel \textit{N-K} SCOPF using adversarial robustness~\cite{agarwal2025large}, and transferability-oriented adversarial robust methods~\cite{zuo2024transferability}. While these approaches have advanced power system security, they primarily rely on predetermined attack scenarios and face computational limitations preventing real-time deployment.

\subsubsection{Reinforcement Learning Approaches for Cyber-Physical Systems}
Reinforcement learning has emerged as a powerful paradigm for enhancing cyber-physical system resilience, particularly in power networks. Pioneering work by \cite{wei2019cyber} introduced deep RL for cyber-attack recovery, followed by \cite{roberts2021deep} who developed specialized approaches for voltage unbalance attacks. \cite{arnold2022adam} advanced Adam-based augmented random search for attack scenarios, while \cite{wu2022reinforcement} introduced physics-inspired graph convolutional networks that leverage power network structure. However, these approaches primarily consider simple random attack scenarios and cannot guarantee system recovery or constraint satisfaction during operation.

To address these limitations, constrained reinforcement learning (CRL) has emerged as a critical framework for power systems requiring strict safety and operational constraints~\cite{su2025review}. Unlike conventional RL that focuses solely on reward maximization, CRL explicitly handles dual objectives of performance optimization while ensuring constraint satisfaction. In power systems, these constraints encompass physical limits (voltage bounds, thermal ratings) and operational requirements (power balance, reserve margins)~\cite{wu2023constrained}. The integration of constraint-handling mechanisms enables RL to navigate complex trade-offs between economic efficiency and system security~\cite{su2025safe}. 
Nevertheless, existing CRL methods face significant challenges under worst-case cyber-physical attack scenarios, often struggling to provide feasible training environments as worst-case scenarios can render constraint sets infeasible, preventing effective policy learning~\cite{coulson2021distributionally, yu2021robust}.

\vspace{-0.3cm}
\subsection{Contributions}
To address these critical challenges, this paper proposes a robust constrained reinforcement learning (RCRL) framework that enables real-time defense against cyber-physical attacks in power systems. The main contributions of this work are:
\begin{enumerate}
\item \textbf{Dynamic Worst-Case Attack Defense:} We develop a tri-level optimization approach that dynamically identifies worst-case \textit{N-K} attack scenarios, coordinating real-time defensive strategies against large-scale coordinated attacks. Unlike static approaches that assume fixed conditions, our method adapts to varying operational states for effective real-time mitigation.
    
\item \textbf{Robust Constrained RL Methodology:} Our method ensures training and deployment safety through novel beta-blending projection-based feasible action mapping during training and primal-dual augmented Lagrangian optimization for deployment. This maintains strict constraint satisfaction under adversarial uncertainty while enabling real-time decision-making.
    
\item \textbf{Theoretical Safety and Robust Guarantees:} We provide rigorous mathematical proofs establishing convergence guarantees and constraint satisfaction bounds, ensuring that learned policies maintain system states within safe regions while achieving near-optimal performance under worst-case attack scenarios.
\end{enumerate}

The remainder of this paper is structured as follows: Section~\ref{sec:problem_formulation} formulates the tri-level optimization framework, Section~\ref{sec:methodology} develops the RCRL approach, Section~\ref{sec:results} provides experimental validation, and Section~\ref{sec:conclusion} concludes.

\begin{figure*}[t]
\centering
\includegraphics[width=0.98\textwidth]{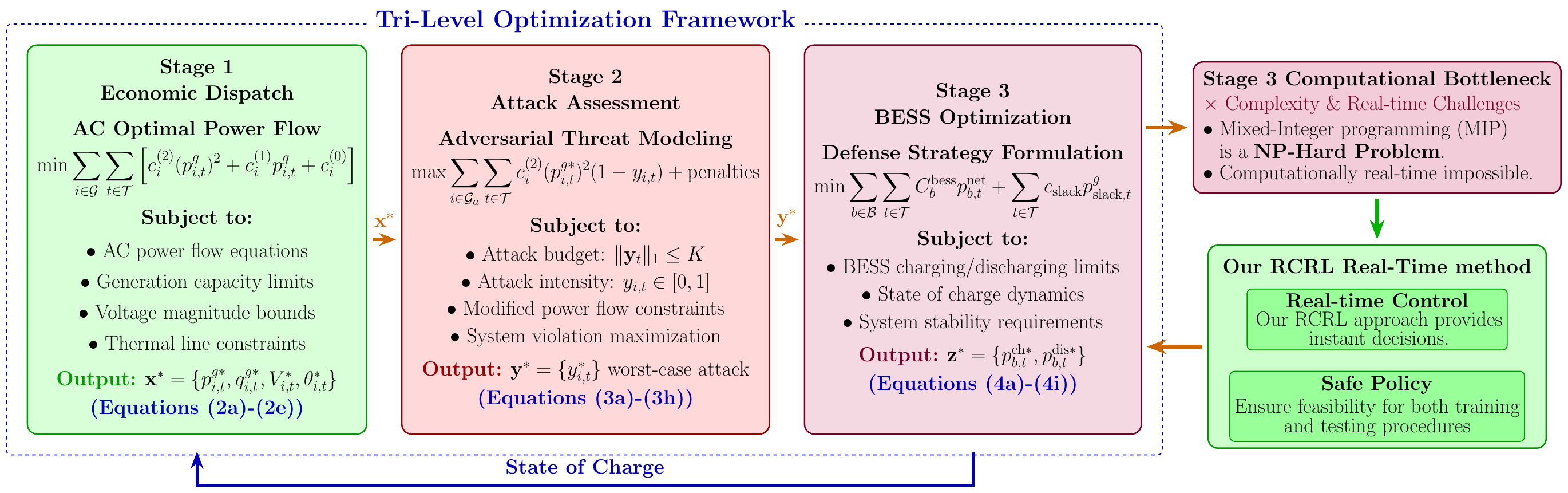}
\caption{Tri-level framework for power system cyber-physical defense}
\label{fig:framework}
\vspace{-0.6cm}
\end{figure*}

\section{Problem Formulation} \label{sec:problem_formulation}

This section formulates the tri-level optimization framework for power system cyber-physical security. As shown in Figure~\ref{fig:framework}, our real-time approach employs three sequential stages: (i) AC-OPF under normal operating conditions to establish baseline generation schedules, (ii) adversarial attack analysis (AAA) to identify worst-case attacks on generators, and (iii) battery energy storage system (BESS) coordination for real-time attack mitigation.

\vspace{-5mm}
\subsection{Mathematical Framework Overview}
We formulate the three-stage problem as a sequential decision process over a power system network characterized by the following sets: $\mathcal{N} = \{1, 2, \ldots, N\}$ represents system buses, $\mathcal{G} = \{1, 2, \ldots, G\}$ denotes generators, and $\mathcal{L} = \{1, 2, \ldots, L\}$ indicates transmission lines, where $\mathcal{G}_i \subseteq \mathcal{G}$ specifies the generators connected to bus $i$. Additionally, $\mathcal{G}_a \subseteq \mathcal{G}$ identifies vulnerable generators susceptible to cyber-physical attacks, $\mathcal{B} = \{1, 2, \ldots, B\}$ represents  BESS, and $\mathcal{T} = \{1, 2, \ldots, T\}$ defines the operational time horizon. %Economic parameters include generator cost coefficients $c_i^{(2)}, c_i^{(1)}, c_i^{(0)} \in \mathbb{R}_+$, slack costs $c_{\mathrm{slack}}^{(2)}, c_{\mathrm{slack}}^{(1)}, c_{\mathrm{slack}}^{(0)} \in \mathbb{R}_+$, and BESS costs $C_b \in \mathbb{R}_+$. 
 
The framework employs three decision vectors: AC-OPF variables $\mathbf{x} := \{p_{i,t}^g, q_{i,t}^g, V_{i,t}, \theta_{i,t}\}$ where $p_{i,t}^g, q_{i,t}^g$ are active/reactive power generation and $V_{i,t}, \theta_{i,t}$ are voltage magnitudes/angles; attack variables $\mathbf{y} := \{y_{i,t}\}$ where $y_{i,t}$ represents attack intensity on generator $i$; and BESS variables $\mathbf{z} := \{p_{b,t}^{\mathrm{ch}}, p_{b,t}^{\mathrm{dis}}, q_{b,t}^{\mathrm{bess}}, \mathrm{soc}_{b,t}\}$ where these denote charging/discharging power, reactive power, and state of charge (SOC). 
The three-stage   optimization framework is:
\begin{subequations} \label{eq:sequential_framework}
\begin{align}
\mathcal{S}_1(\mathcal{X}): \quad &\mathbf{x}^* \leftarrow \arg\min_{\mathbf{x} \in \mathcal{X}} J_1(\mathbf{x}) \label{eq:alg_stage1} \\
\mathcal{S}_2(\mathbf{x}^*): \quad &\mathbf{y}^* \leftarrow \arg\max_{\mathbf{y} \in \mathcal{Y}} J_2(\mathbf{x}^*, \mathbf{y}) \label{eq:alg_stage2} \\
\mathcal{S}_3(\mathbf{x}^*, \mathbf{y}^*): \quad &\mathbf{z}^* \leftarrow \arg\min_{\mathbf{z} \in \mathcal{Z}} J_3(\mathbf{x}^*, \mathbf{y}^*, \mathbf{z}) \label{eq:alg_stage3}
\end{align}
\end{subequations}
where $J_1(\mathbf{x})$, $J_2(\mathbf{x}^*, \mathbf{y})$, and $J_3(\mathbf{x}^*, \mathbf{y}^*, \mathbf{z})$ represent economic dispatch cost minimization, attack impact maximization, and defense cost minimization, respectively.

\vspace{-5mm}

\subsection{Stage 1: Economic Dispatch with AC Optimal Power Flow}\label{sec:OPF_normal_condition}
The first stage establishes baseline generation dispatch by solving an AC-OPF problem under normal operating conditions, corresponding to $\mathcal{S}_1$ in \eqref{eq:sequential_framework}. The system is characterized by power demands $P_{i,t}^d, Q_{i,t}^d \in \mathbb{R}_+$ at each bus, admittance matrix elements $G_{ij}, B_{ij} \in \mathbb{R}$ where $Y_{ik} = G_{ik} + jB_{ik}$ with magnitude $\abs{Y_{ik}}$ and angle $\alpha_{ik}$, and economic parameters including generator cost coefficients $c_i^{(2)}, c_i^{(1)}, c_i^{(0)} \in \mathbb{R}_+$. 

Operational constraints are defined by generation limits $p_i^{g,\min}, p_i^{g,\max} \in \mathbb{R}_+$ and $q_i^{g,\min}, q_i^{g,\max} \in \mathbb{R}$, voltage bounds $V_i^{\min}, V_i^{\max} \in \mathbb{R}_+$, and thermal capacity limits $S_{ij}^{\max} \in \mathbb{R}_+$ for transmission lines. The apparent power flow $S_{ij,t} \in \mathbb{C}$ represents the complex power flow on line $(i,j)$. The AC-OPF formulation minimizes the total generation cost as follows:
\begin{subequations} \label{mod:stage1_opf}
\setlength{\jot}{2pt}
\begin{align}
&\min_{\mathbf{x}} J_1(\mathbf{x}) = \sum_{i \in \mathcal{G}} \sum_{t \in \mathcal{T}} \left(c_i^{(2)}(p_{i,t}^g)^2 + c_i^{(1)}p_{i,t}^g + c_i^{(0)}\right) \label{obj:stage1_cost} \\
&\text{subject to:} \nonumber \\
&\sum_{j \in \mathcal{G}_i} p_{j,t}^g - P_{i,t}^d = P_{i,t}^{\text{flow}}(V, \theta) \label{cons:stage1_P_balance}\\
&\sum_{j \in \mathcal{G}_i} q_{j,t}^g - Q_{i,t}^d = Q_{i,t}^{\text{flow}}(V, \theta) \label{cons:stage1_Q_balance}\\
&p_i^{g,\min} \leq p_{i,t}^g \leq p_i^{g,\max}, q_i^{g,\min} \leq q_{i,t}^g \leq q_i^{g,\max} \label{cons:stage1_gen_bounds} \\
&V_i^{\min} \leq V_{i,t} \leq V_i^{\max}, |S_{ij,t}| \leq S_{ij}^{\max} \label{cons:stage1_network_limits}\\
&\forall i \in \mathcal{N}, j \in \mathcal{G}, (i,j) \in \mathcal{L}, t \in \mathcal{T} \nonumber
\end{align}
\end{subequations}
where the AC power flow functions $P_{i,t}^{\text{flow}}(V, \theta)$ and $Q_{i,t}^{\text{flow}}(V, \theta)$ represent feasible network power injections determined by voltage variables as defined in equations \eqref{cons:ac_pflow}--\eqref{cons:ac_qflow}.
\begin{subequations} \label{eq:ac_power_flow_set}
\begin{align}
P_{i,t}^{\text{flow}}(V, \theta) &:= V_{i,t} \sum_{k \in \mathcal{N}} V_{k,t} \abs{Y_{ik}} \cos(\theta_{i,t} - \theta_{k,t} - \alpha_{ik}) \label{cons:ac_pflow} \\
Q_{i,t}^{\text{flow}}(V, \theta) &:= V_{i,t} \sum_{k \in \mathcal{N}} V_{k,t} \abs{Y_{ik}} \sin(\theta_{i,t} - \theta_{k,t} - \alpha_{ik}) \label{cons:ac_qflow}
\end{align}
\end{subequations}
The feasible set $\mathcal{X}$ is defined by constraints \eqref{cons:stage1_P_balance}--\eqref{cons:stage1_network_limits}. Constraints \eqref{cons:stage1_P_balance} and \eqref{cons:stage1_Q_balance} enforce Kirchhoff's current law by requiring that net power injections at each bus satisfy the AC power flow equations. Constraint \eqref{cons:stage1_gen_bounds} enforces generator operational limits, while constraint \eqref{cons:stage1_network_limits} ensures voltage operating ranges and thermal capacity limits are satisfied. The solution yields an optimal dispatch $\mathbf{x}^*$ that establishes the baseline operating point for subsequent attack assessment.

\subsection{Stage 2: Adversarial Attack Assessment (AAA) Model} \label{sec:AAA_model}
\subsubsection{The Worst-Case Attack Model}
We consider sophisticated adversaries capable of launching coordinated cyber-physical attacks against power generation infrastructure. These attacks involve compromising generator control systems to forcibly reduce power output, with adversaries dynamically selecting target generators based on real-time system conditions and optimal dispatch information from Stage 1. The attack strategy adapts continuously to maximize system-wide disruption given the current operational state, with primary objectives including escalating operational costs through forced dispatch adjustments, violating voltage magnitude constraints at critical buses, and overloading transmission lines beyond their thermal capacity limits.
Under the worst-case scenario, adversaries maintain persistent control over compromised generators without immediate detection or intervention, enabling them to execute sustained adaptive attacks that exploit time-varying system vulnerabilities. This dynamic threat model captures the most sophisticated attack scenarios where adversaries continuously adjust their targeting strategy based on evolving system conditions. By designing our defensive framework against such intelligent, adaptive adversaries, we ensure robust protection against a wide spectrum of real-world cyber-physical threats.

\subsubsection{Attack Impact Maximization Formulation}
Building on the worst-case attack model, we define the attack decision vector $\mathbf{y} := \{y_{i,t}\}$ with continuous variables $y_{i,t} \in [0,1]$ representing attack intensity on generator $i$ at time $t$, targeting generators in set $\mathcal{G}_a \subseteq \mathcal{G}$ co-located with BESS units. Following coordinated attacks, swing bus generation can typically provide rapid power adjustments to counteract frequency deviations from sudden generator losses~\cite{zhao2012swing}. However, the resulting system state may violate voltage bounds and power flow limit constraints, rendering them infeasible. Consequently, we relax generator bound constraints~\eqref{cons:stage1_gen_bounds} and network limit constraints~\eqref{cons:stage1_network_limits} in the subsequent formulation.

To streamline the mathematical presentation, we define the slack generation cost function as 
\begin{equation}
C^{\text{slk}}(p^g_{\mathrm{slack},t}) := c_{\mathrm{slack}}^{(2)}(p^g_{\mathrm{slack},t})^2 + c_{\mathrm{slack}}^{(1)}p^g_{\mathrm{slack},t} + c_{\mathrm{slack}}^{(0)}, \label{slack_cost}
\end{equation}
where $c_{\mathrm{slack}}^{(2)}, c_{\mathrm{slack}}^{(1)}, c_{\mathrm{slack}}^{(0)} \in \mathbb{R}_+$ are the corresponding cost coefficients. The adversary solves the following optimization problem to maximize system disruption:
\begin{subequations} \label{mod:AAA}
\setlength{\jot}{0pt}
\begin{align}
&\max_{\mathbf{y}} J_2(\mathbf{x}^*, \mathbf{y}) = \nonumber \\
\begin{split}
&\sum_{i \in \mathcal{G}_a} \sum_{t \in \mathcal{T}} \left(c_i^{(2)}(p_{i,t}^{g*}(1-y_{i,t}))^2 + c_i^{(1)}p_{i,t}^{g*}(1-y_{i,t}) + c_i^{(0)}\right) \\
&+  \sum_{t \in \mathcal{T}} C^{\text{slk}}(p^g_{\mathrm{slack},t}) + {\xi}_1 \inf_{(i,j) \in \mathcal{L}, t \in \mathcal{T}} \Psi_{i,j,t}   + {\xi}_2 \inf_{i \in \mathcal{N}, t \in \mathcal{T}} \Omega_{i,t}
\end{split} \label{obj:obj_attack} \\
&\text{subject to:}  \nonumber \\
\begin{split}
&\sum_{j \in \mathcal{G}_i} p_{j,t}^{g*}(1-y_{j,t}) + \delta_i^{\mathrm{slack}} p^g_{\mathrm{slack},t} - P_{i,t}^d  = P_{i,t}^{\text{flow}}(V, \theta)
\end{split} \label{cons:attack_p_balance} \\
\begin{split}
&\sum_{j \in \mathcal{G}_i} q_{j,t}^{g*}(1-y_{j,t}) + \delta_i^{\mathrm{slack}} q^g_{\mathrm{slack},t} - Q_{i,t}^d   = Q_{i,t}^{\text{flow}}(V, \theta)
\end{split} \label{cons:attack_q_balance} \\
\begin{split}
&\Psi_{i,j,t} = \max\{S_{i,j,t} - S_{ij}^{\max}, -S_{i,j,t} - S_{ij}^{\max}, 0\}
\end{split} \label{cons:attack_line_violations} \\
&\Omega_{i,t} = \max\{V_{i,t} - V_i^{\max}, V_i^{\min} - V_{i,t}, 0\} \label{cons:attack_voltage_violations} \\
\begin{split}
&p_{\mathrm{slack}}^{\min} \leq p^g_{\mathrm{slack},t} \leq p_{\mathrm{slack}}^{\max},  q_{\mathrm{slack}}^{\min} \leq q^g_{\mathrm{slack},t} \leq q_{\mathrm{slack}}^{\max}
\end{split} \label{cons:slack_bounds} \\
&\sum_{i \in \mathcal{G}_a} y_{i,t} \leq K \label{cons:attack_max_simultaneous}\\
&\forall i \in \mathcal{N}, j \in \mathcal{G}, (i,j) \in \mathcal{L}, t \in \mathcal{T} \nonumber
\end{align}
\end{subequations}
The feasible set $\mathcal{Y}$ is defined by constraints \eqref{cons:attack_p_balance}--\eqref{cons:attack_max_simultaneous}, where the objective function \eqref{obj:obj_attack} maximizes system disruption through economic impact from removing low-cost generation and forcing expensive slack generation, plus operational stress via constraint violation penalties weighted by parameters $\xi_1, \xi_2 \in \mathbb{R}_+$ for transmission line and voltage violations, respectively. The coordination constraint \eqref{cons:attack_max_simultaneous} limits the total attack intensity to $K$, representing practical coordination limitations in multi-target cyber-physical attack scenarios.

Power flow constraints \eqref{cons:attack_p_balance}--\eqref{cons:attack_q_balance} incorporate attack effects through reduction factors $(1-y_{j,t})$, while constraints \eqref{cons:attack_line_violations} and \eqref{cons:attack_voltage_violations} define line flow violation measures $\Psi_{i,j,t}$ and voltage violation measures $\Omega_{i,t}$ to quantify maximum constraint violations across the system. This formulation yields the worst-case attack vector $\mathbf{y}^*$ for Stage 3 defensive  control.

\vspace{-3mm}
\subsection{Stage 3: BESS-Based Attack Mitigation}\label{sec:Attack_min}

The third stage implements defensive actions through fast-response BESS to mitigate the worst-case attacks identified in Stage 2, corresponding to $\mathcal{S}_3$ in \eqref{eq:sequential_framework}. Given the fixed attack strategy $\mathbf{y}^*$ from Stage 2, the defender determines an optimal BESS control strategy $\mathbf{z}^*$ to minimize operational costs while maintaining system constraints.

The BESS framework employs several decision variables: charging and discharging power $p_{b,t}^{\mathrm{ch}}, p_{b,t}^{\mathrm{dis}} \geq 0$, reactive power injection $q_{b,t}^{\mathrm{bess}} \in \mathbb{R}$, and SOC $\mathrm{soc}_{b,t} \in [0,1]$. Binary variables $\beta_{b,t}^{\mathrm{ch}}, \beta_{b,t}^{\mathrm{dis}} \in \{0,1\}$ enforce mutually exclusive charging and discharging modes, while the given connectivity matrix $\mu_{i,b} \in \{0,1\}$ indicates whether BESS unit $b$ connects to bus $i$. The net power injection $p_{b,t}^{\mathrm{net}} := p_{b,t}^{\mathrm{dis}} - p_{b,t}^{\mathrm{ch}}$ is positive for discharging and negative for charging.
Key operational parameters include charging/discharging efficiencies $\eta_b^{\mathrm{ch}}, \eta_b^{\mathrm{dis}} \in (0,1]$, maximum energy capacity $E_b^{\max} \in \mathbb{R}_+$, power limits $p_b^{\mathrm{ch,max}}, p_b^{\mathrm{dis,max}} \in \mathbb{R}_+$, reactive power bounds $q_b^{\mathrm{bess,min}}, q_b^{\mathrm{bess,max}} \in \mathbb{R}$, SOC limits $\mathrm{soc}_b^{\min}, \mathrm{soc}_b^{\max} \in [0,1]$, and time step duration $\Delta t \in \mathbb{R}_+$.  The BESS operational costs are characterized by coefficients $C_b \in \mathbb{R}_+$.
The defense cost minimization objective $J_3(\mathbf{x}^*, \mathbf{y}^*, \mathbf{z})$ is formulated as:
\begin{subequations} \label{mod:stage3_bess}
\setlength{\jot}{0pt}
\begin{align}
&\min_{\mathbf{z}} J_3(\mathbf{x}^*, \mathbf{y}^*, \mathbf{z}) = \sum_{t \in \mathcal{T}} C^{\text{slk}}(p^g_{\mathrm{slack},t}) + \sum_{b \in \mathcal{B}} \sum_{t \in \mathcal{T}} C_b p_{b,t}^{\mathrm{net}} \nonumber \\
&\quad + \xi_1 \sup_{(i,j) \in \mathcal{L}, t \in \mathcal{T}} \Psi_{i,j,t} + \xi_2 \sup_{i \in \mathcal{N}, t \in \mathcal{T}} \Omega_{i,t} \label{obj:stage3_cost} \\
&\text{subject to:}  \nonumber \\
\begin{split}
&\sum_{j \in \mathcal{G}_i} p_{j,t}^{g*}(1-y_{j,t}^*) + \sum_{b \in \mathcal{B}} \mu_{i,b} p_{b,t}^{\mathrm{net}} \\
&\quad + \delta_i^{\mathrm{slack}} p^g_{\mathrm{slack},t} - P_{i,t}^d = P_{i,t}^{\text{flow}}(V, \theta)
\end{split} \label{cons:stage3_p_balance} \\
\begin{split}
&\sum_{j \in \mathcal{G}_i} q_{j,t}^{g*}(1-y_{j,t}^*) + \sum_{b \in \mathcal{B}} \mu_{i,b} q_{b,t}^{\mathrm{bess}} \\
&\quad + \delta_i^{\mathrm{slack}} q^g_{\mathrm{slack},t} - Q_{i,t}^d = Q_{i,t}^{\text{flow}}(V, \theta)
\end{split} \label{cons:stage3_q_balance} \\
&p_{b,t}^{\mathrm{net}} = p_{b,t}^{\mathrm{dis}} - p_{b,t}^{\mathrm{ch}} \label{cons:bess_net_power} \\
\begin{split}
&\mathrm{soc}_{b,t} = \mathrm{soc}_{b,t-1} + \frac{\eta_b^{\mathrm{ch}} p_{b,t}^{\mathrm{ch}} - p_{b,t}^{\mathrm{dis}}/\eta_b^{\mathrm{dis}}}{E_b^{\max}} \Delta t
\end{split} \label{cons:soc_dynamics} \\
&0 \leq p_{b,t}^{\mathrm{ch}} \leq \beta_{b,t}^{\mathrm{ch}} p_b^{\mathrm{ch,max}}, \, 0 \leq p_{b,t}^{\mathrm{dis}} \leq \beta_{b,t}^{\mathrm{dis}} p_b^{\mathrm{dis,max}} \label{cons:bess_power_limits} \\
&q_b^{\mathrm{bess,min}} \leq q_{b,t}^{\mathrm{bess}} \leq q_b^{\mathrm{bess,max}}, \, \mathrm{soc}_b^{\min} \leq \mathrm{soc}_{b,t} \leq \mathrm{soc}_b^{\max} \label{cons:bess_limits} \\
&p_{\mathrm{slack}}^{\min} \leq p^g_{\mathrm{slack},t} \leq p_{\mathrm{slack}}^{\max}, \, q_{\mathrm{slack}}^{\min} \leq q^g_{\mathrm{slack},t} \leq q_{\mathrm{slack}}^{\max} \label{cons:stage3_slack_bounds} \\
&\eqref{cons:attack_line_violations}\text{--}\eqref{cons:attack_voltage_violations}, \, \beta_{b,t}^{\mathrm{ch}} + \beta_{b,t}^{\mathrm{dis}} \leq 1 \label{cons:stage3_violations} \\
&\forall i \in \mathcal{N}, b \in \mathcal{B}, (i,j) \in \mathcal{L}, t \in \mathcal{T} \nonumber
\end{align}
\end{subequations}
The feasible set $\mathcal{Z}$ is defined by constraints \eqref{cons:stage3_p_balance}--\eqref{cons:stage3_violations}, where the objective function \eqref{obj:stage3_cost} minimizes total system operational costs through three primary components: slack generation costs, BESS operational costs proportional to net power injections, and constraint violation penalties weighted by parameters $\xi_1, \xi_2 \in \mathbb{R}_+$ for transmission line and voltage violations, respectively. Unlike Stage 2, which maximizes violations, Stage 3 minimizes total system costs, including violation penalties during defense operations.

The formulation incorporates fixed attack effects $\mathbf{y}^*$ from Stage 2 through power flow constraints \eqref{cons:stage3_p_balance}--\eqref{cons:stage3_q_balance}, while BESS power injections are integrated via the connectivity matrix. Constraints \eqref{cons:bess_net_power}--\eqref{cons:soc_dynamics} define BESS operational dynamics, including energy balance with charging/discharging efficiencies, whereas constraints \eqref{cons:bess_power_limits}--\eqref{cons:bess_limits} enforce operational limits and mutually exclusive charging/discharging modes. Finally, constraint \eqref{cons:stage3_violations} references violation measures from Stage 2 and enforces slack generator limits. This formulation yields optimal BESS control strategy $\mathbf{z}^* = \{p_{b,t}^{\mathrm{ch}*}, p_{b,t}^{\mathrm{dis}*}, q_{b,t}^{\mathrm{bess}*}\}$ for real-time emergency response.

\section{Methodology} \label{sec:methodology}
Emergency grid control requires BESS coordination on a sub-second timescale. While the Stage 3 model addresses this need, its mixed-integer complexity prohibits traditional solvers that require seconds to minutes, during which cascading failures can cause complete system collapse. This timing mismatch makes conventional optimization operationally unviable for preventing fast-propagating grid emergencies.

To address this challenge, we propose an RCRL approach   that learns optimal BESS coordination policies offline, enabling real-time response during attacks. The power system control problem is modeled as a Constrained Markov Decision Process (CMDP) where the state space $\mathcal{S}$ captures post-attack system conditions, the action space $\mathcal{A}$ represents BESS control decisions, and the reward function $\mathcal{R}$ reflects the Stage 3 cost minimization objective from \eqref{obj:stage3_cost}. The key innovation integrates power system constraints \eqref{cons:stage3_p_balance}--\eqref{cons:stage3_violations} directly into policy learning, ensuring feasible emergency response within milliseconds across diverse scenarios.

%==============================
\begin{figure*}[!t]
\centering
\includegraphics[width=0.94\textwidth]{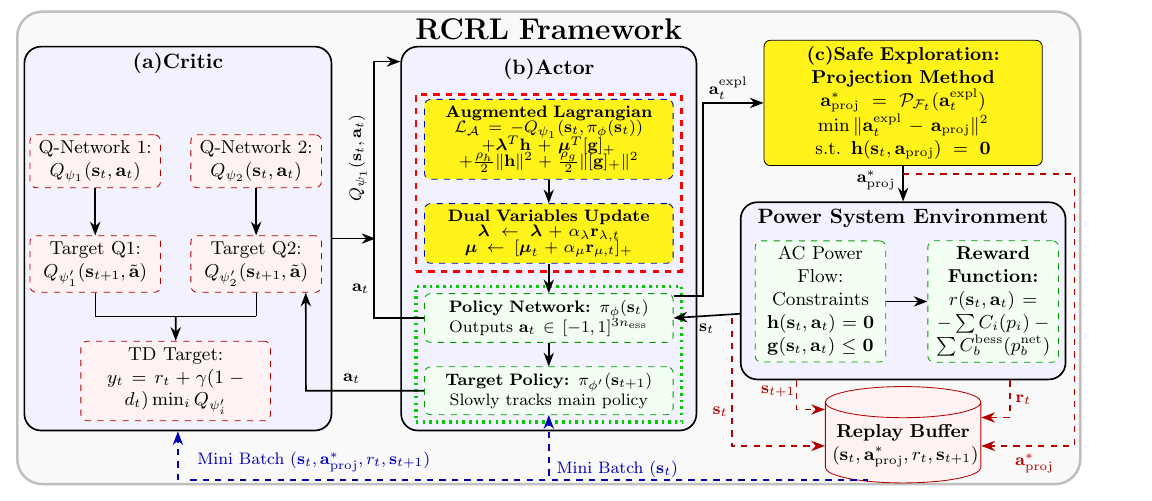}
\caption{RCRL architecture: (a) dual critic networks, (b)  actor with constraint handling, (c) projection-based safe exploration.}
\label{fig:constrained_td3_framework}
\vspace{-0.6cm}
\end{figure*}
\subsection{Actor-Critic Architecture for BESS Control} \label{sec:actor-critic}
The RCRL framework extends the Twin Delayed Deep Deterministic Policy Gradient (TD3) algorithm~\cite{fujimoto2018addressing} by incorporating constraint handling and robustness mechanisms for safe BESS control under cyber-physical attacks, as illustrated in Fig.~\ref{fig:constrained_td3_framework}. The enhanced architecture maintains TD3's dual critic networks while adding augmented Lagrangian optimization for constraint handling and projection-based safe exploration to ensure operational feasibility throughout the learning process.
%===============================
%
\subsubsection{State and Action Space Design}
The state space captures essential post-attack system conditions for BESS control:
\begin{equation} \label{eq:state_space_rl}
\mathbf{s}_t := [\mathbf{V}_t^T, \boldsymbol{\theta}_t^T, (\mathbf{P}_t^{\text{inj}})^T, \mathbf{soc}_t^T]^T \in \mathbb{R}^{n_s}
\end{equation}
where $\mathbf{V}_t \in \mathbb{R}^N$ and $\boldsymbol{\theta}_t \in \mathbb{R}^N$ represent voltage magnitudes and angles at all buses, $\mathbf{P}_t^{\text{inj}} \in \mathbb{R}^N$ captures net power injections, and $\mathbf{soc}_t \in \mathbb{R}^B$ tracks the state-of-charge for all BESS units. The state dimension is $n_s := 3N + B$.

The action space controls BESS operations with normalized values for stable learning as shown in equation \eqref{eq:action_space_rl}.
\begin{equation} \label{eq:action_space_rl}
\mathbf{a}_t := [(\mathbf{a}_t^{\text{ch}})^T, (\mathbf{a}_t^{\text{dis}})^T, (\mathbf{a}_t^{Q})^T]^T \in [-1,1]^{n_a}
\end{equation}
where $\mathbf{a}_t^{\text{ch}}, \mathbf{a}_t^{\text{dis}}, \mathbf{a}_t^{Q} \in \mathbb{R}^B$ represent normalized charging, discharging, and reactive power controls, with $n_a = 3B$ total action dimension.

Actions are mapped by model \eqref{mod:action_mapping_rl} to physical values through standard affine transformations. 
The $(\mathbf{a}_t + 1)/2$ transformation maps the normalized RL action space $[-1,1]$ to physically valid BESS operating ranges. For power actions, $-1$ corresponds to zero power (battery off) and $+1$ to maximum power, ensuring non-negative charging/discharging powers while maintaining stable RL training with standardized action bounds \eqref{cons:bess_power_limits}. For reactive power, the transformation scales the full allowable range from minimum to maximum values \eqref{cons:bess_limits}.
\begin{subequations} \label{mod:action_mapping_rl}
\begin{align}
p_{b,t}^{\mathrm{ch}} &:= \frac{p_b^{\mathrm{ch,max}}}{2}(\mathbf{a}_t^{\text{ch}}[b] + 1), \quad p_{b,t}^{\mathrm{dis}} := \frac{p_b^{\mathrm{dis,max}}}{2}(\mathbf{a}_t^{\text{dis}}[b] + 1) \label{eq:P_mapping_rl} \\
q_{b,t}^{\mathrm{bess}} &:= q_b^{\mathrm{bess,min}} + \frac{q_b^{\mathrm{bess,max}} - q_b^{\mathrm{bess,min}}}{2}(\mathbf{a}_t^{Q}[b] + 1) \label{eq:Q_mapping_rl}
\end{align}
\end{subequations}

\subsubsection{Reward Function and System Integration}
The reward function serves as the critical bridge between the RL framework and the Stage 3 power system optimization objective from \eqref{obj:stage3_cost}, ensuring the learned policy respects power system constraints while optimizing operational economics. The reward function is defined as the negative total operational cost:
\begin{equation} \label{eq:reward_function_rl}
r(\mathbf{s}_t,\mathbf{a}_t) := -\left(\sum_{b \in \mathcal{B}} C_b p_{b,t}^{\mathrm{net}} + C^{\text{slk}}(p^g_{\mathrm{slack},t}) + C^{\text{viol}}(\mathbf{s}_t,\mathbf{a}_t)\right)
\end{equation}
where the reward function comprises three cost components that directly align with the Stage 3 objective. The first term $\sum_{b \in \mathcal{B}} C_b p_{b,t}^{\mathrm{net}}$ represents BESS operational costs, where $p_{b,t}^{\mathrm{net}}$ denotes the net power injection from battery unit $b$ at time $t$. The second term is defined in \eqref{slack_cost} that captures slack generation costs for balancing power demand. The third term $C^{\text{viol}}(\mathbf{s}_t,\mathbf{a}_t) = \xi_1\sup_{k=1,\dots,4} \Phi^k_{i,j,t} + \xi_2 \sup_{k=5,6} \Phi^k_{i,t} $ penalizes constraint violations to encourage feasible operation. All power variables are computed from AC power flow solutions that incorporate BESS control actions through the mappings in \eqref{mod:action_mapping_rl}.

\subsubsection{Actor Network Architecture}
The actor network implements a deterministic policy $\pi_\phi: \mathcal{S} \rightarrow \mathcal{A}$ that maps power system states to BESS control actions:
\begin{equation} \label{eq:actor_policy_rl}
\pi_\phi(\mathbf{s}_t) = \tanh(f_\phi(\mathbf{s}_t)) \in [-1,1]^{3B}
\end{equation}
where $f_\phi(\mathbf{s}_t)$ represents a three-layer feedforward network with ReLU activations and the $\tanh$ function ensures normalized action bounds.

\subsubsection{Critic Network Architecture}
The TD3 algorithm employs two independent critic networks ($i=1,2$) to address overestimation bias. Each critic network $Q_{\psi_i}: \mathcal{S} \times \mathcal{A} \rightarrow \mathbb{R}$ estimates action-values:
\begin{equation} \label{eq:critic_networks_rl}
Q_{\psi_i}(\mathbf{s}_t, \mathbf{a}_t) = h_{\psi_i}([\mathbf{s}_t; \mathbf{a}_t])
\end{equation}
where $[\mathbf{s}_t; \mathbf{a}_t]$ denotes concatenated state-action inputs and $h_{\psi_i}$ represents the $i$-th critic function using a standard three-layer feedforward architecture with ReLU activations.

\subsubsection{Target Networks and Training Dynamics}

TD3 employs target networks $\pi_{\phi'}$ and $Q_{\psi_i'}$ with slowly evolving parameters to ensure training stability. As defined in equation \eqref{eq:td_target_rl}, the target Q-value computation incorporates clipped double Q-learning with target policy smoothing.
\begin{equation} \label{eq:td_target_rl}
y_t = r_t + \gamma(1-d_t)\min\{Q_{\psi_1'}(\mathbf{s}_{t+1}, \tilde{\mathbf{a}}), Q_{\psi_2'}(\mathbf{s}_{t+1}, \tilde{\mathbf{a}})\}
\end{equation}
where $\tilde{\mathbf{a}} = \text{clip}(\pi_{\phi'}(\mathbf{s}_{t+1}) + \boldsymbol{\epsilon}, -1, 1)$ with clipped Gaussian noise $\boldsymbol{\epsilon}$, discount factor $\gamma \in (0,1)$, and terminal flag $d_t \in \{0,1\}$.
The critic networks minimize the Bellman error $\mathcal{L}_{\text{critic}}(\psi_i) = \mathbb{E}[(y_t - Q_{\psi_i}(\mathbf{s},\mathbf{a}))^2]$, while the actor maximizes the first critic's Q-value: $\mathcal{L}_{\text{actor}}(\phi) = -\mathbb{E}[Q_{\psi_1}(\mathbf{s}, \pi_\phi(\mathbf{s}))]$. Target networks are updated via Polyak averaging with rate $\tau = 0.005$.

\subsection{Safe Exploration: Beta-Blending Projection Method} \label{sec:safe-exploration}

In safety-critical power systems, unconstrained RL exploration can cause constraint violations leading to voltage instability or cascading failures. We propose a Beta-Blending projection approach that combines safe projection with gradual exploration to maintain both safety and learning efficiency. Unlike \cite{wu2023constrained}, which only handles feasibility during testing, our approach ensures safety throughout the entire training process. When infeasible actions are generated during training, our method ensures constraint satisfaction and continues learning, whereas approaches that lack training-phase safety mechanisms may encounter constraint violations that could halt the training process.

%============================================================

\subsubsection{Feasible Set and Projection Formulation}
For system state $\mathbf{s}_t$, the feasible action set $\mathcal{F}_t \subset [-1,1]^{3B}$ is defined as:
\begin{equation} \label{eq:feasible_set_definition}
\mathcal{F}_t := \{\mathbf{a} \in [-1,1]^{3B} : \mathbf{h}(\mathbf{s}_t, \mathbf{a}) = \mathbf{0}, \mathbf{g}(\mathbf{s}_t, \mathbf{a}) \leq \mathbf{0}\}
\end{equation}
where $\mathbf{h}(\mathbf{s}_t, \mathbf{a}): \mathbb{R}^{n_s} \times \mathbb{R}^{3B} \rightarrow \mathbb{R}^{m_e}$ represents equality constraints (power balance and SOC dynamics) and $\mathbf{g}(\mathbf{s}_t, \mathbf{a}): \mathbb{R}^{n_s} \times \mathbb{R}^{3B} \rightarrow \mathbb{R}^{m_i}$ represents inequality constraints (voltage bounds, line limits, and BESS operational constraints).

The projection operator $\mathcal{P}_{\mathcal{F}_t}: \mathbb{R}^{3B} \rightarrow \mathcal{F}_t$ finds the closest feasible action by solving:
\begin{equation}\label{eq:projection_formulation}
\mathcal{P}_{\mathcal{F}_t}(\mathbf{a}_t^{\text{expl}}) := \arg\min_{\mathbf{a}^{\text{proj}} \in \mathcal{F}_t} \|\mathbf{a}_t^{\text{expl}} - \mathbf{a}^{\text{proj}}\|^2
\end{equation}

\subsubsection{Beta-Blending Innovation}
Our approach introduces a time-varying blending parameter $\beta_t \in [0,1]$ that gradually transitions from safe to optimal exploration with schedule $\beta_t := \min(t/T_{\beta}, 1)$:
\begin{equation} \label{eq:beta_blending}
\mathbf{a}_t^{\text{final}} := \beta_t \mathbf{a}_t^{\text{expl}} + (1-\beta_t)\mathcal{P}_{\mathcal{F}_t}(\mathbf{a}_t^{\text{expl}})
\end{equation}
% with schedule:
% \begin{equation} \label{eq:beta_schedule}
% \beta_t := \min(t/T_{\beta}, 1)
% \end{equation}

\begin{remark}[Beta-Blending: Resolving the Safety-Learning Paradox]
Traditional approaches face a fundamental paradox: relying solely on projection during training prevents neural networks from learning constraint satisfaction since the projection operator corrects violations externally, creating projection dependence. Conversely, using only primal-dual methods without projection violates constraints during early training and may destabilize learning. Beta-blending resolves this paradox through a novel graduated handoff mechanism: initially ($\beta_t \approx 0$), projection ensures safety while primal-dual gradients train network parameters to internalize constraints. As $\beta_t$ increases, the approach gradually shifts from external correction to learned constraint awareness. By late training ($\beta_t \approx 1$), the network generates inherently safe actions without projection, eliminating computational overhead while maintaining constraint satisfaction—critical for millisecond emergency response.
\end{remark}

\subsection{Safe Exploitation: Primal-Dual Method} \label{sec:safe-exploitation}
The Beta-Blending approach ensures safety during training, but real-world power system deployment demands sub-millisecond decision-making that renders online constraint projection computationally prohibitive. To bridge this gap, we develop a primal-dual augmented Lagrangian framework that embeds constraint enforcement within the Beta-Blending TD3 framework, enabling the neural network to internalize feasibility and eliminate the need for runtime projection during deployment.

\subsubsection{Constrained Policy Optimization Framework}
We reformulate Stage 3 from \eqref{mod:stage3_bess} as a constrained reinforcement learning problem to enable real-time BESS coordination under adversarial conditions, as presented in \eqref{eq:constrained_policy_optimization}.
\begin{subequations}\label{eq:constrained_policy_optimization}
\begin{align}
\max_{\pi_{\phi}} \quad &J(\phi) = \mathbb{E}_{\tau \sim \pi_{\phi}} \left[ \sum_{t=0}^{T-1} \gamma^t r(\mathbf{s}_t, \mathbf{a}_t) \right] \label{eq:policy_objective}\\
\text{s.t.} \quad & \mathbb{P}\left[\mathbf{h}(\mathbf{s}_t, \mathbf{a}_t)  = \mathbf{0} \right] = 1 - \delta_h, \quad \forall t \in \{0, \dots, T-1\} \label{eq:equality_constraints_policy}\\
& \mathbb{P}\left[\mathbf{g}(\mathbf{s}_t, \mathbf{a}_t) \leq \mathbf{0}\right]  = 1 - \delta_g, \quad \forall t \in \{0, \dots, T-1\} \label{eq:inequality_constraints_policy}
\end{align}
\end{subequations}
where $\tau := \{(\mathbf{s}_0, \mathbf{a}_0, r_0), (\mathbf{s}_1, \mathbf{a}_1, r_1), \ldots, (\mathbf{s}_{T-1}, \mathbf{a}_{T-1}, r_{T-1})\}$ denotes a state-action-reward trajectory under policy $\pi_{\phi}$, and $\lim_{T\to \infty}\delta_h, \delta_g \to 0$ represent constraint violation tolerances that diminish as training progresses. The stochasticity arises from exploration noise $\boldsymbol{\epsilon}_t$ and diverse attack scenarios generated by the tri-level framework. Constraints $\mathbf{h}(\mathbf{s}_t, \mathbf{a}_t)$ and $\mathbf{g}(\mathbf{s}_t, \mathbf{a}_t)$  represent power flow equations, generation limits, voltage bounds, and thermal constraints. The probabilistic satisfaction requirements arise from stochastic policy $\pi_{\phi}$, ensuring constraints are satisfied almost surely as training converges. This reformulation transforms the computationally prohibitive mixed-integer optimization from Stage 3 into a tractable policy learning framework, enabling millisecond emergency response while maintaining operational constraints.

%====== Original ==========

%==============================
\subsubsection{Augmented Lagrangian Formulation}
To solve the constrained policy optimization problem in \eqref{eq:constrained_policy_optimization}, we transform it into an unconstrained problem using augmented Lagrangian techniques. Since the original objective $J(\phi)$ in \eqref{eq:policy_objective} represents the expected cumulative reward, we reformulate this at the policy gradient level by replacing the instantaneous reward maximization with Q-function maximization for each time step. This allows us to incorporate constraint penalties directly into the temporal difference learning framework.
The augmented Lagrangian for a single time step is formulated as:
\begin{equation} \label{eq:augmented_lagrangian_vectorized}
\begin{aligned}
\mathcal{L}_{\mathcal{A}}(\phi, \boldsymbol{\lambda}_t, \boldsymbol{\mu}_t) &= -Q_{\psi}(\mathbf{s}_t, \pi_{\phi}(\mathbf{s}_t)) + \boldsymbol{\lambda}_t^T \mathbf{h}(\mathbf{s}_t, \pi_{\phi}(\mathbf{s}_t)) \\
&+ \boldsymbol{\mu}_t^T [\mathbf{g}(\mathbf{s}_t, \pi_{\phi}(\mathbf{s}_t))]_+ \\
&+ \frac{\rho}{2}\left(\|\mathbf{h}(\mathbf{s}_t, \pi_{\phi}(\mathbf{s}_t))\|^2 + \|[\mathbf{g}(\mathbf{s}_t, \pi_{\phi}(\mathbf{s}_t))]_+\|^2\right)
\end{aligned}
\end{equation}
where $Q_{\psi}(\mathbf{s}_t, \mathbf{a}_t)$ is the critic network defined in \eqref{eq:actor_policy_rl}, $\boldsymbol{\lambda}_t \in \mathbb{R}^{m_e}$ and $\boldsymbol{\mu}_t \in \mathbb{R}_+^{m_i}$ are Lagrange multipliers for equality constraints $\mathbf{h}(\mathbf{s}_t, \mathbf{a}_t)$ from \eqref{eq:equality_constraints_policy} and inequality constraints $\mathbf{g}(\mathbf{s}_t, \mathbf{a}_t)$ from \eqref{eq:inequality_constraints_policy}, respectively, $[\cdot]_+ = \max(0, \cdot)$ denotes the positive part operator, and $\rho > 0$ is the penalty parameter that controls constraint violation penalties.

The dual variables are updated using standard augmented Lagrangian procedures:
\begin{subequations} \label{eq:dual_updates}
\begin{align}
\boldsymbol{\lambda}_{t+1} &= \boldsymbol{\lambda}_t + \alpha_{\lambda} \mathbf{r}_{\lambda,t} \label{eq:lambda_update}\\
\boldsymbol{\mu}_{t+1} &= [\boldsymbol{\mu}_t + \alpha_{\mu} \mathbf{r}_{\mu,t}]_+ \label{eq:mu_update} 
\end{align}
\end{subequations}
where $\alpha_{\lambda}, \alpha_{\mu} > 0$ are dual learning rates, $\mathbf{r}_{\lambda,t} := \mathbf{h}(\mathbf{s}_t, \pi_{\phi}(\mathbf{s}_t))$ represents the equality constraint residuals, and $\mathbf{r}_{\mu,t} := [\mathbf{g}(\mathbf{s}_t, \pi_{\phi}(\mathbf{s}_t))]_+$ represents the inequality constraint residuals.

The primal update optimizes the actor policy parameters by minimizing the augmented Lagrangian:
\begin{equation} \label{eq:primal_update}
\phi \leftarrow \phi - \eta_{\phi} \nabla_{\phi} \mathcal{L}_{\mathcal{A}}(\phi, \boldsymbol{\lambda}_t, \boldsymbol{\mu}_t)
\end{equation}
where $\eta_{\phi} > 0$ is the primal learning rate.
This framework enables real-time constraint satisfaction by embedding all Stage 3 physics-based constraints directly into policy learning, eliminating the need for online optimization during emergency response scenarios.

%========== Original ====================

\subsubsection{Primal-Dual Update Algorithm}\label{sec:primal-dual_update}
The primal-dual algorithm alternates between policy parameter updates and dual variable updates. The policy parameters evolve according to gradient descent on the augmented Lagrangian with clipping bounds $\Lambda_{\max}, M_{\max} > 0$ to maintain stability as specified in \eqref{eq:practical_dual_updates}.
The clipping bounds are chosen to prevent dual variables from growing excessively during training while preserving effective constraint enforcement \cite{fujimoto2018addressing}.
\begin{subequations}
\label{eq:practical_dual_updates}
\begin{align}
\boldsymbol{\lambda}_t^{(k+1)} &= \text{clip}\left(\boldsymbol{\lambda}_t^{(k)} + \alpha_{\lambda} \mathbf{r}_{\lambda,t}^{(k)}, -\Lambda_{\max}, \Lambda_{\max}\right) \label{eq:lambda_update_clip} \\
\boldsymbol{\mu}_t^{(k+1)} &= \text{clip}\left([\boldsymbol{\mu}_t^{(k)} + \alpha_{\mu} \mathbf{r}_{\mu,t}^{(k)}]_+, 0, M_{\max}\right) \label{eq:mu_update_clip}
\end{align}
\end{subequations}
where $k$ denotes the iteration index, and dual variables are updated every $\kappa \in \mathbb{Z}_+$ primal iterations to maintain computational efficiency while ensuring constraint enforcement. This algorithm achieves asymptotic constraint satisfaction under standard regularity conditions.

\subsection{Theoretical Convergence Analysis}
Parallel to our previous work [8], we develop theorems below for the tri-level RCRL framework. The main differences are: (i) \textbf{Multi-stage attack-dependent constraints}: Unlike [8] which addresses standard OPF constraints, our theorems handle cascaded constraints where Stage 3 feasibility depends on attack scenarios from Stages 1-2; (ii) \textbf{Beta-blending convergence analysis}: The convergence analysis incorporates the time-varying $\beta$ parameter that transitions from projection-based to direct policy actions; (iii) \textbf{Attack-resilient constraint structure}: We establish Linear Independence Constraint Qualification  (LICQ) for power systems under coordinated \textit{N-K}attacks rather than stochastic demand uncertainty.

The key theoretical advancement is demonstrating that Beta-blending projection operations do not compromise asymptotic convergence while enabling safe exploration during training.

\begin{theorem}[Convergence to Stationary Points]
\label{thm:primal_dual_convergence}
Let $\{(\phi^{(k)}, \boldsymbol{\lambda}_t^{(k)}, \boldsymbol{\mu}_t^{(k)})\}$ denote the sequence generated by \eqref{eq:primal_update} and \eqref{eq:practical_dual_updates}. Under the assumptions that (i) $\pi_{\phi}$ is twice continuously differentiable with Lipschitz continuous gradients, (ii) constraint functions $h_j(\mathbf{s}_t, \mathbf{a}_t)$ and $g_i(\mathbf{s}_t, \mathbf{a}_t)$ satisfy the Linear Independence Constraint Qualification, (iii) policy parameter space $\Phi$ is compact, and (iv) dual variables remain bounded through projection operations, every accumulation point $(\phi^*, \boldsymbol{\lambda}_t^*, \boldsymbol{\mu}_t^*)$ satisfies the first-order stationarity condition in \eqref{eq:stationarity_condition}.
\begin{equation}
\label{eq:stationarity_condition}
\nabla_{\phi} \mathcal{L}_{\mathcal{A}}(\phi^*, \boldsymbol{\lambda}_t^*, \boldsymbol{\mu}_t^*) \in \mathcal{N}_{\Phi}(\phi^*)
\end{equation}
where $\mathcal{N}_{\Phi}(\phi^*)$ denotes the normal cone to $\Phi$ at $\phi^*$.
\end{theorem}

\begin{corollary}[Beta-blending Projection Impact]
\label{cor:beta_blending_convergence}
Under the Beta-blending framework with schedule $\beta_t = \min(t/T_{\beta}, 1)$, projection operations do not affect asymptotic convergence properties. For $t \geq T_{\beta}$ where $\beta_t = 1$, the stationarity condition in \eqref{eq:stationarity_condition} reduces to the standard augmented Lagrangian case since projection effects vanish.
\end{corollary}

%=============original========================

%
%=======================================

\begin{theorem}[Constraint Violation Bounds]
\label{thm:constraint_satisfaction}
Under the conditions of Theorem~\ref{thm:primal_dual_convergence} and assuming the constrained MDP satisfies Slater's constraint qualification (there exists $(\mathbf{s}, \mathbf{a}) \in \mathcal{F}$ such that $\mathbf{g}(\mathbf{s}, \mathbf{a}) < \mathbf{0}$ strictly), for any $\epsilon > 0$, there exist penalty parameter thresholds $\rho_h^{\min}, \rho_g^{\min} > 0$ such that for penalty parameters $\rho_h \geq \rho_h^{\min}$ and $\rho_g \geq \rho_g^{\min}$ in the augmented Lagrangian \eqref{eq:augmented_lagrangian_vectorized}, every stationary point $(\phi^*, \boldsymbol{\lambda}_t^*, \boldsymbol{\mu}_t^*)$ of $\mathcal{L}_{\mathcal{A}}$ satisfies the constraint violation bounds in \eqref{eq:constraint_bounds}.
\begin{equation}
\label{eq:constraint_bounds}
\|\mathbf{r}_{\lambda,t}\| \leq \epsilon \quad \text{and} \quad \|\mathbf{r}_{\mu,t}\| \leq \epsilon
\end{equation}
The penalty parameter thresholds are computed using \eqref{eq:penalty_thresholds}.
\begin{subequations}
\label{eq:penalty_thresholds}
\begin{align}
\rho_h^{\min} &= \frac{\|\nabla_{\phi} Q_{\psi_1}(\mathbf{s}_t, \pi_{\phi^*}(\mathbf{s}_t))\| + \Lambda_{\max}}{\epsilon} \label{eq:penalty_h} \\
\rho_g^{\min} &= \frac{M_{\max}}{\epsilon} \label{eq:penalty_g}
\end{align}
\end{subequations}
where $\rho_h$ and $\rho_g$ correspond to the penalty parameter $\rho$ in \eqref{eq:augmented_lagrangian_vectorized} for equality and inequality constraints, respectively. This establishes convergence to policies satisfying power system constraints within tolerance $\epsilon$.
\end{theorem}

\vspace{-0.7cm}
\subsection{Algorithm Summary} \label{sec:algorithm}

This section presents the integrated algorithmic framework that combines our tri-level optimization approach with RCRL. The methodology operates across two distinct phases: an offline training phase, where the agent learns optimal BESS coordination strategies through safe exploration via the Beta-Blending Projected Primal-Dual Method, and an online deployment phase that enables real-time emergency response without requiring projection operations.

The complete algorithmic framework is presented in Algorithm~\ref{alg:trilevel_constrained_td3}, which integrates all three optimization stages. Stage 1 solves AC-OPF from \eqref{mod:stage1_opf} for baseline dispatch (line 6), Stage 2 implements AAA from \eqref{mod:AAA} to identify worst-case attacks (line 7), and Stage 3 deploys RCRL for BESS coordination (lines 8-11). The Beta-Blending mechanism transitions from safe exploration ($\beta_t = 0$) to safe exploitation ($\beta_t = 1$) as training progresses.

Safe exploration occurs in lines 9-11 through projection-based feasible action generation per \eqref{eq:projection_formulation} and Beta-blending per \eqref{eq:beta_blending} to ensure constraint satisfaction during training. Safe exploitation occurs in lines 22-24 through augmented Lagrangian actor optimization per \eqref{eq:augmented_lagrangian_vectorized} and dual variable updates per \eqref{eq:practical_dual_updates} (lines 18-19) that embed constraints into policy learning for deployment without projection. The augmented Lagrangian embeds all Stage 3 constraints into policy learning, enabling millisecond response times without online optimization.

\begin{algorithm}[!htbp]
\caption{Tri-Level Constrained TD3 with Beta-Blending}
\label{alg:trilevel_constrained_td3}
\begin{algorithmic}[1]
\STATE \textbf{Init:} $\phi, \psi_1, \psi_2$, buffer $\mathcal{D}$, $\boldsymbol{\lambda}, \boldsymbol{\mu}$, global\_step $= 0$
\STATE \textbf{Set:} $\tau = 0.005$, $\gamma = 0.99$, $\eta_a = \eta_c = 3 \times 10^{-4}$, $\alpha_{\lambda} = \alpha_{\mu} = 0.5$
\FOR{episode $e = 1$ to $E_{\max}$}
   \STATE Reset env., observe $\mathbf{s}_0$, init. $\mathbf{SOC}_0 = 0.9$
   \FOR{timestep $t = 0$ to $T-1$}
       \STATE \textbf{Stage 1:} $\mathbf{x}^* \leftarrow \arg\min J_1(\mathbf{x})$ \eqref{mod:stage1_opf}
       \STATE \textbf{Stage 2:} $\mathbf{y}^* \leftarrow \arg\max J_2(\mathbf{x}^*, \mathbf{y})$ \eqref{mod:AAA}
       \STATE \textbf{Stage 3:} $\beta_t \leftarrow \min(\text{global\_step}/10^5, 1)$
       \STATE Gen. action: $\mathbf{a}_t^{\text{expl}} = \pi_\phi(\mathbf{s}_t) + \boldsymbol{\epsilon}_t$
       \STATE Project: $\mathbf{a}_t^{\text{proj}} = \mathcal{P}_{\mathcal{F}_t}(\mathbf{a}_t^{\text{expl}})$
       \STATE Beta-blend: $\mathbf{a}_t^{\text{final}} = \beta_t \mathbf{a}_t^{\text{expl}} + (1-\beta_t)\mathbf{a}_t^{\text{proj}}$
       \STATE Execute \& store: $(\mathbf{s}_t, \mathbf{a}_t^{\text{final}}, r_t, \mathbf{s}_{t+1}, d_t)$ in $\mathcal{D}$
       \IF{$|\mathcal{D}| \geq B_{\text{min}}$}
          \STATE Sample batch $\{(\mathbf{s}_i, \mathbf{a}_i, r_i, \mathbf{s}_i', d_i)\}_{i=1}^{N}$ from $\mathcal{D}$
          \STATE \textbf{Critic:} Compute TD targets, update $\psi_i$
          \STATE \textbf{Constraints:} $\mathbf{r}_{\lambda,t} = \mathbf{h}(\mathbf{s}_t, \mathbf{a}_t)$, $\mathbf{r}_{\mu,t} = [\mathbf{g}(\mathbf{s}_t, \mathbf{a}_t)]_+$
          \IF{$t \bmod 10 = 0$}
              \STATE $\boldsymbol{\lambda} \leftarrow \text{clip}(\boldsymbol{\lambda} + \alpha_{\lambda} \mathbf{r}_{\lambda,t}, -\Lambda_{\max}, \Lambda_{\max})$
              \STATE $\boldsymbol{\mu} \leftarrow \text{clip}([\boldsymbol{\mu} + \alpha_{\mu} \mathbf{r}_{\mu,t}]_+, 0, M_{\max})$
          \ENDIF
          \IF{$t \bmod 2 = 0$}
              \STATE \textbf{Actor:} $\mathcal{L} = -Q_{\psi_1}(\mathbf{s}, \pi_\phi(\mathbf{s})) + \boldsymbol{\lambda}^T \mathbf{r}_{\lambda,t} + \boldsymbol{\mu}^T \mathbf{r}_{\mu,t} + \frac{\rho}{2}(\|\mathbf{r}_{\lambda,t}\|^2 + \|\mathbf{r}_{\mu,t}\|^2)$
              \STATE Update: $\phi \leftarrow \phi - \eta_a \nabla_\phi \mathcal{L}$
              \STATE Soft update: $\phi' \leftarrow \tau\phi + (1-\tau)\phi'$, $\psi'_i \leftarrow \tau\psi_i + (1-\tau)\psi'_i$
          \ENDIF
       \ENDIF
       \STATE Update: $\mathbf{s}_t \leftarrow \mathbf{s}_{t+1}$, $\mathbf{SOC}_t \leftarrow \mathbf{SOC}_{t+1}$, global\_step $\leftarrow$ global\_step $+ 1$
   \ENDFOR
\ENDFOR
\STATE \textbf{Return:} $\pi_\phi^*$
\end{algorithmic}
\end{algorithm}

\section{Experimental Results} \label{sec:results}
This section provides comprehensive validation of the proposed tri-level RCRL framework for cyber-physical defense. We evaluate the methodology across multiple dimensions: learning convergence and policy structure analysis, performance assessment including optimality and constraint satisfaction, critical attack scenario analysis, and computational efficiency evaluation. 

% The experiments demonstrate the framework's effectiveness in providing real-time defense against coordinated cyber-physical attacks while maintaining operational constraints.
\vspace{-7mm}
\subsection{Experimental Setup} \label{sec:exp-setup}

The framework was evaluated on IEEE 30-bus and 57-bus test systems using MATPOWER \cite{Zimmerman2011}, with five BESS units strategically deployed at buses 2, 13, 22, 23, and 27 for the 30-bus system and four BESS units at buses 3, 6, 8, and 12 for the 57-bus system. These locations were selected through network topology analysis to maximize defensive coverage against coordinated attacks and provide optimal system-wide coordination against \textit{N-K} attack scenarios. Each BESS features 1000 MWh capacity, 98\% round-trip efficiency, and 30-80 MW power ratings, operating over a 24-hour window with realistic daily load profiles from NREL data \cite{NRELData}. The adversarial threat model simulates a strategic attacker capable of simultaneously compromising up to $K=4$ generators for both test systems. 

% The attacker solves an optimization problem to identify outages that maximize system disruption through generation redispatch costs and constraint violation penalties, creating coordinated attack scenarios that significantly exceed traditional \textit{N-1} security analysis scope.

\subsubsection{CRL Agent Implementation and Training}

The agent was implemented using PyTorch 2.1.0 with an actor-critic architecture featuring two fully-connected hidden layers (256 neurons each) with ReLU activations. Tri-level optimization subproblems were solved using Pyomo with IPOPT interior-point solver. Experiments were conducted on Ubuntu 6.11.0 with Intel i9-14900K (32 cores), 125.5 GB RAM, and NVIDIA RTX 6000 Ada GPU (CUDA 12.1).

The agent trained over 200,000 iterations with hyperparameter tuning: batch size 64, learning rate $3 \times 10^{-4}$ (Adam optimizer), policy delay $d=2$, discount factor $\gamma=0.99$, Polyak averaging $\tau=0.005$, and dual variable updates every 10 iterations. Hyperparameter selection required approximately 15 preliminary training runs over 3 days to identify optimal values, focusing on balancing convergence speed with constraint satisfaction performance. Training required 27.6 hours, after which the policy was evaluated on 5,000 unseen test scenarios.
\begin{figure}
\centering
\includegraphics[width=0.9\columnwidth]{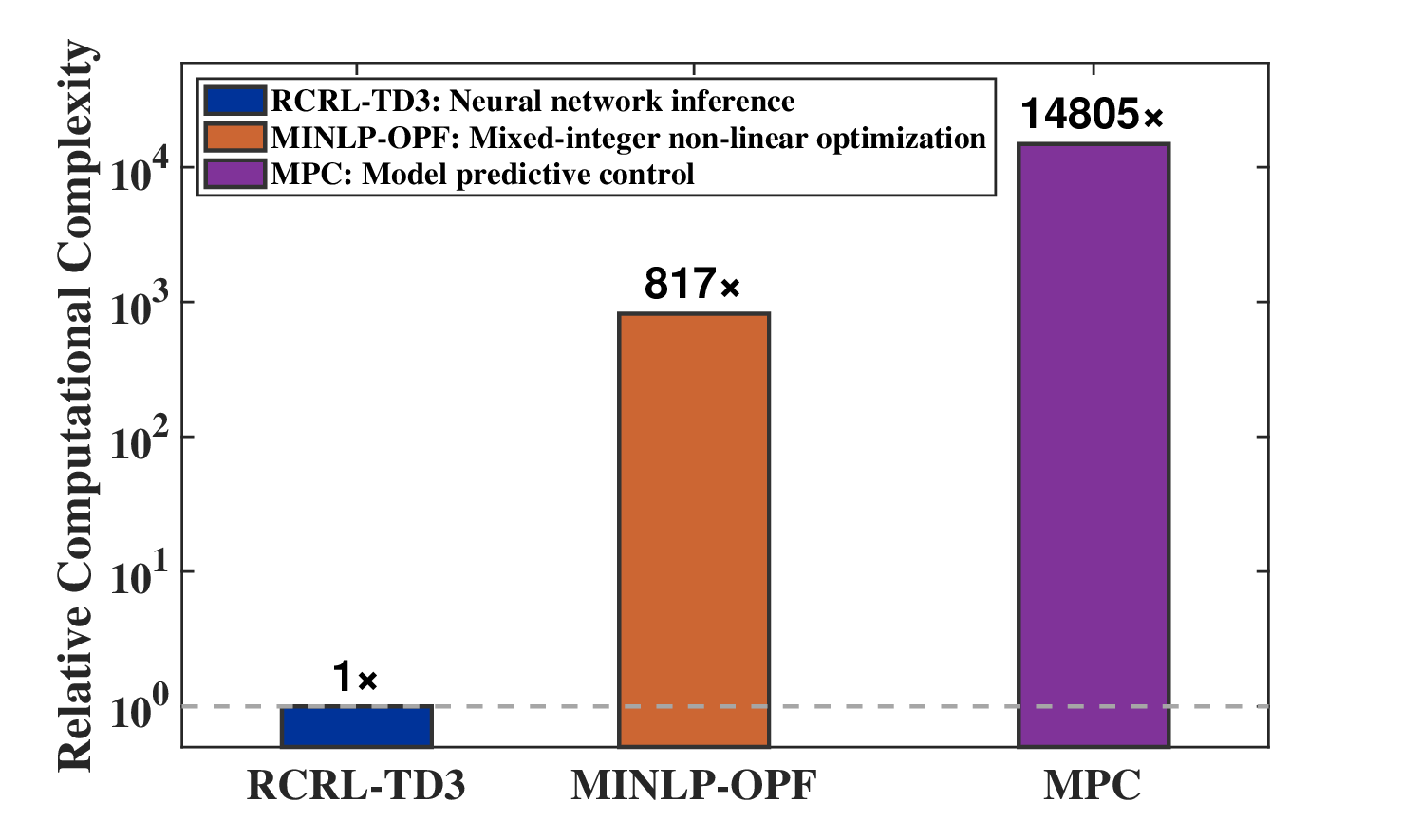}
\caption{Computational performance comparison of real-time control algorithms showing relative complexity factors.}
\label{fig:computation_comparison}
\vspace{-0.3cm}
\end{figure}
\subsubsection{Baseline Methods and Computational Comparison}
We compare three Stage 3 BESS coordination approaches: (i) the proposed constrained TD3 (RCRL-TD3), (ii) direct Stage 3 OPF solving \eqref{mod:stage3_bess} via IPOPT (MINLP-OPF), and (iii) Model Predictive Control (MPC) with a 5-timestep horizon using warm-start IPOPT.

Figure~\ref{fig:computation_comparison} illustrates the computational performance comparison across the three control methodologies on a logarithmic scale. The proposed RCRL-TD3 framework serves as the baseline with a complexity factor of 1×, demonstrating neural network inference capability that enables real-time emergency response. The MINLP-OPF approach exhibits an 817× computational complexity overhead due to its mixed-integer nonlinear optimization requirements, while the MPC controller demonstrates the highest computational burden at 14,805× complexity factor.

The dramatic performance differential highlights the fundamental advantage of the RCRL approach. This computational efficiency is critical for cyber-physical defense applications where millisecond response times are essential to prevent cascading failures during coordinated attacks.

\vspace{-2mm}
\subsection{Training Convergence and Performance Analysis} \label{sec:performance}
\subsubsection{Training Convergence Analysis}

Figure~\ref{fig:learning_curves_30}(a) shows the constrained TD3 agent's training progression over 200,000 iterations (x-axis scaled for readability) for the IEEE 30-bus system, exhibiting three distinct phases: rapid initial learning (0--50,000 iterations) where the agent discovers fundamental BESS control principles, policy refinement (50,000--150,000 iterations) with gradual improvement and reduced variance, and convergence stabilization (150,000--200,000 iterations) where reward plateau indicates successful optimization. The cyclical reward patterns reflect adaptation to varying load conditions and attack scenarios, demonstrating the algorithm's capability to handle dynamic operational environments. Figure~\ref{fig:learning_curves_57}(a) demonstrates similar convergence behavior for the IEEE 57-bus system, with comparable learning phases and stable convergence characteristics.

\subsubsection{Performance Gap Analysis}
The constrained TD3 agent achieves 100\% constraint satisfaction across all 500 test scenarios while delivering real-time control decisions within 0.21 ms per state. Figures~\ref{fig:learning_curves_30}(b) and \ref{fig:learning_curves_57}(b) illustrate the test reward curves and performance gap analysis for both test systems. The performance gap, quantified as:
\begin{equation}
\text{Gap}(\%) := \frac{|\text{Reward} - \text{ACOPF Cost}|}{|\text{ACOPF Cost}|} \times 100
\end{equation}
averages only 5.2\% for the IEEE 30-bus system and 2.04\% for the IEEE 57-bus system across all scenarios, demonstrating that guaranteed constraint satisfaction incurs minimal economic penalty. The key breakthrough is the simultaneous achievement of perfect constraint compliance and sub-millisecond computation time, enabling practical deployment in cyber-physical attack scenarios where traditional optimization methods fail due to computational delays.

\begin{figure}
\centering
\begin{subfigure}[b]{\linewidth}
\centering
    \includegraphics[width=0.9\textwidth]{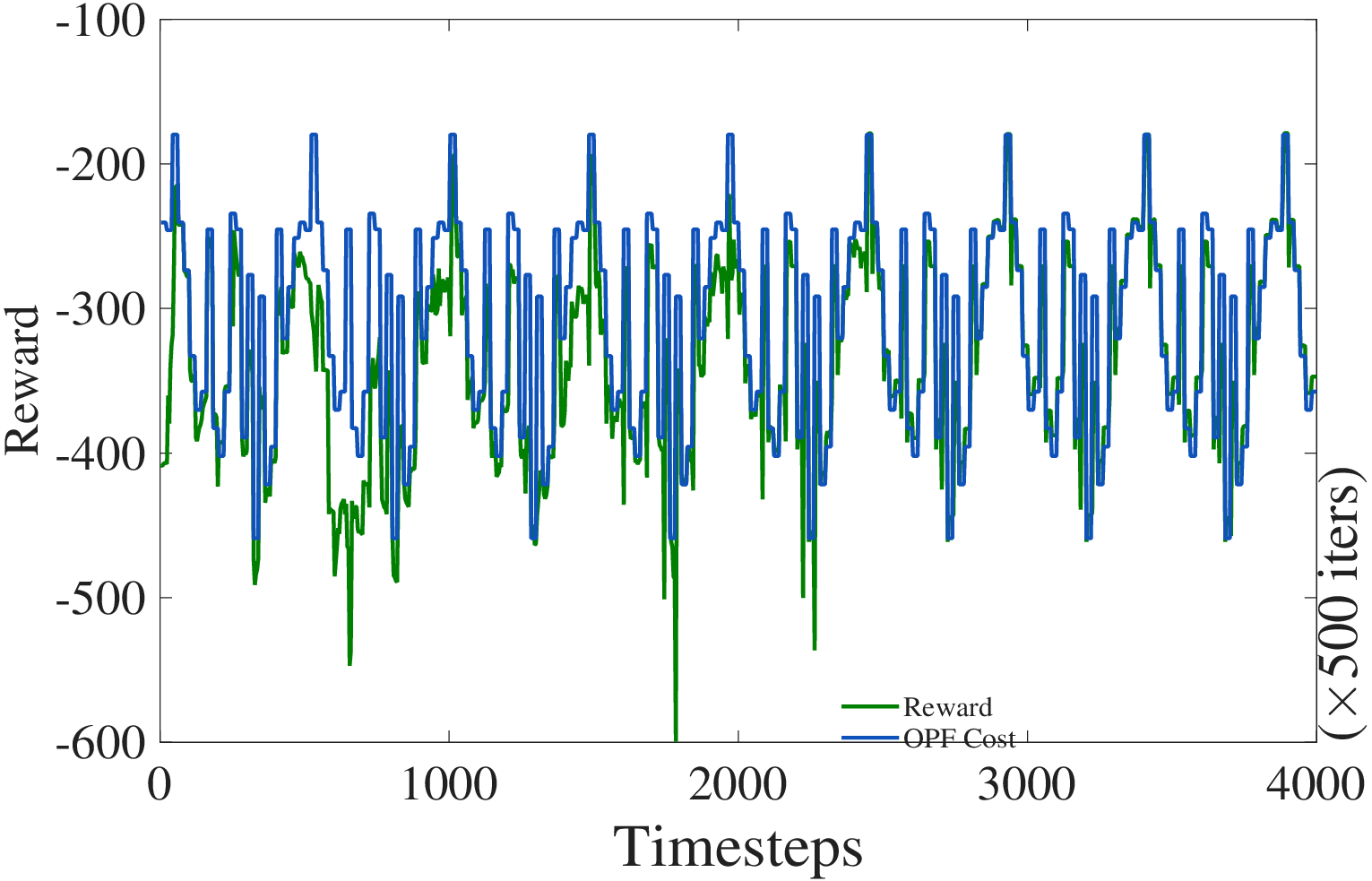}
    \caption{Training Reward Curve}
    \label{fig:train_reward_30}
\end{subfigure}
\vspace{-0.2cm}
\begin{subfigure}[b]{\linewidth}
\centering
    \includegraphics[width=0.9\textwidth]{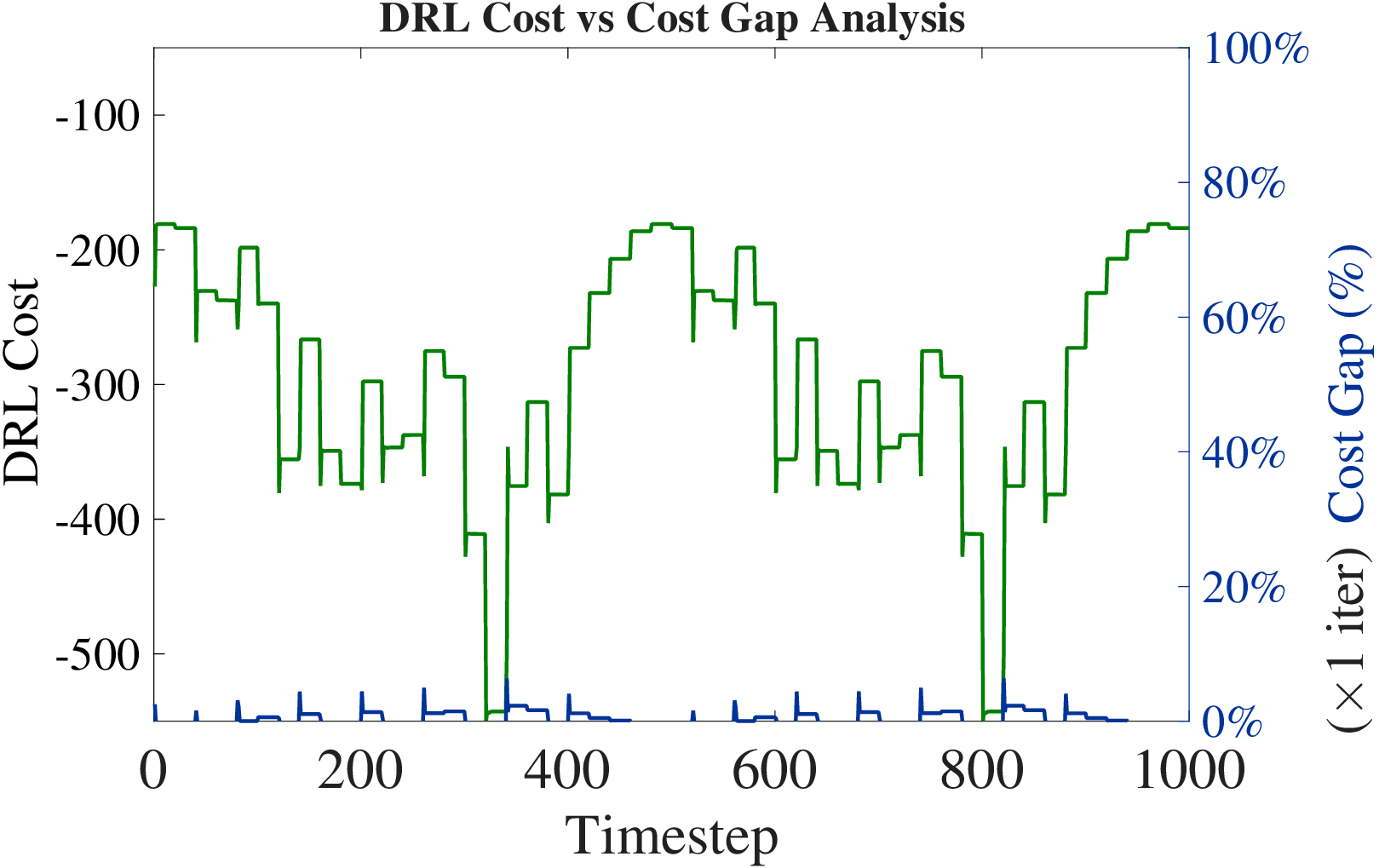}
    \caption{Test Reward (\textcolor{Green}{---} left axis) vs Cost Gap (\textcolor{Abi}{---} right axis)  with 5.20\% average optimality gap and 6.97\% peak deviations.}
    \label{fig:gap_30}
\end{subfigure}
\caption{CRL agent performance for IEEE-30 bus test system.}
\label{fig:learning_curves_30}
\vspace{-0.6cm}
\end{figure}
\begin{figure}
\centering
\begin{subfigure}[b]{\linewidth}
\centering
    \includegraphics[width=0.9\textwidth]{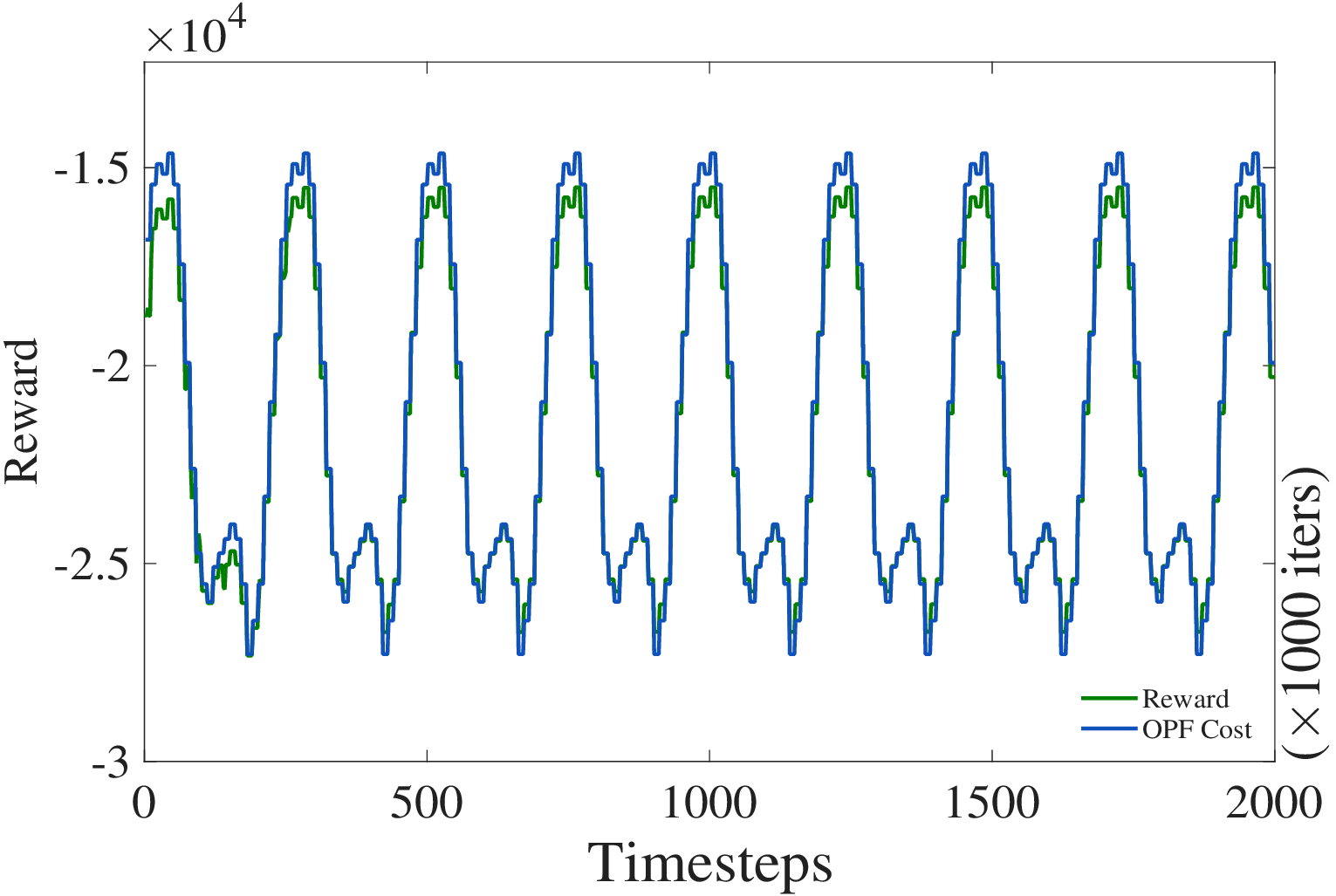}
    \caption{Training Reward Curve}
    \label{fig:train_reward_57}
\end{subfigure}
\vspace{-0.2cm}
\begin{subfigure}[b]{\linewidth}
\centering
% \vspace{-0.06cm}
    \includegraphics[width=0.9\textwidth]{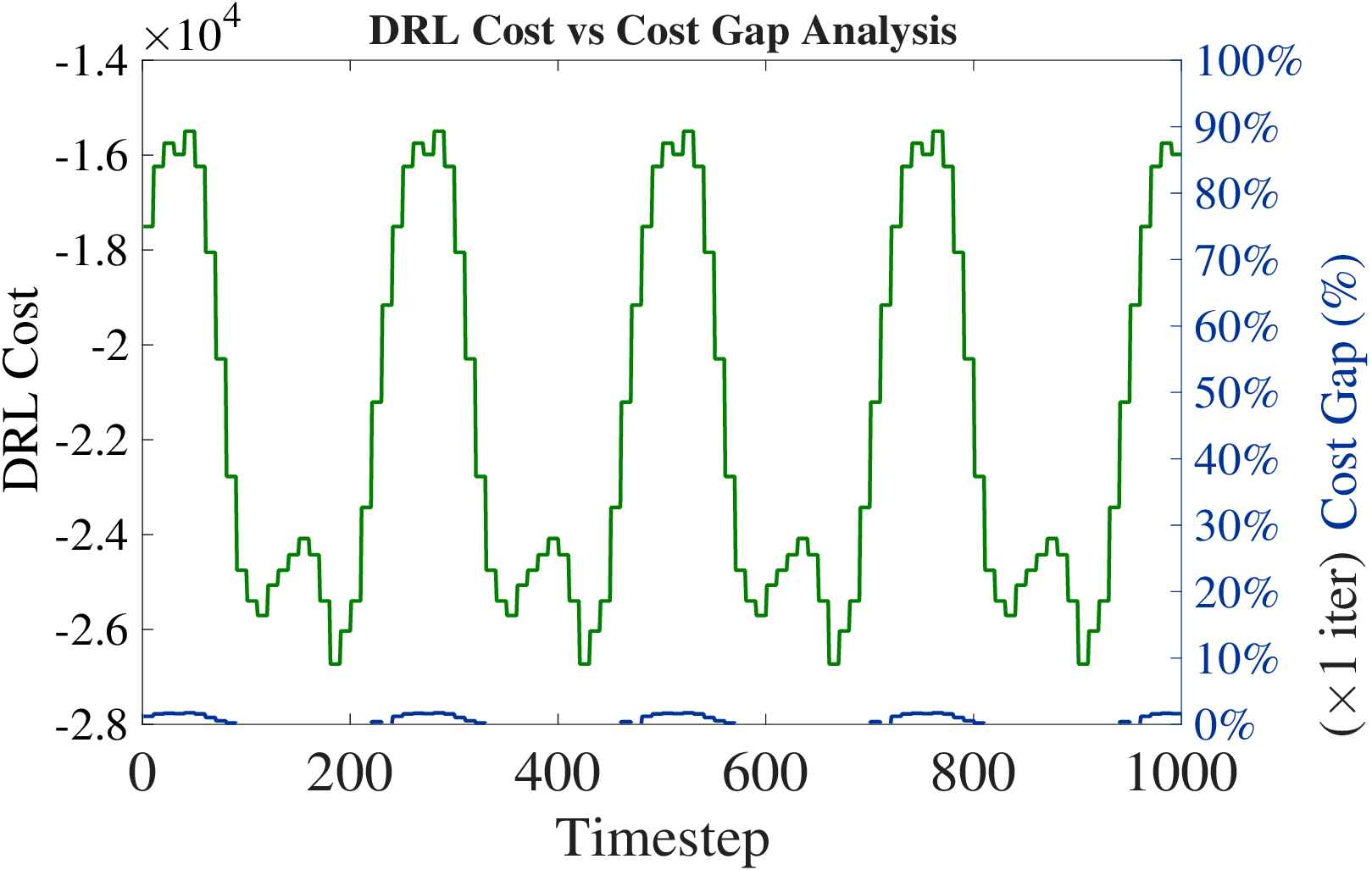}
    \caption{Test Reward (\textcolor{Green}{---} left axis) vs Cost Gap (\textcolor{Abi}{---} right axis) with 2.04\% average optimality gap and 5.86\% peak deviations.}
    \label{fig:gap_57}
\end{subfigure}
\caption{CRL agent performance for IEEE 57-bus test system.}
\label{fig:learning_curves_57}
\vspace{-0.6cm}
\end{figure}

\subsection{Coordinated Attack Response Analysis} \label{sec:case-study}
This case study validates the tri-level RCRL framework through analysis of coordinated cyber-physical attack scenarios that exceed traditional \textit{N-1} contingency capabilities. Figures~\ref{fig:trilevel_validation_compact} and \ref{fig:trilevel_validation_compact_line} demonstrate the complete framework validation, showing how the approach addresses the critical two-stage bottleneck: eliminating computational delays in worst-case scenario identification and defensive response coordination.

\subsubsection{Voltage Response Analysis (Fig. \ref{fig:trilevel_validation_compact})}
Stage 1 AC-OPF optimization from \eqref{mod:stage1_opf} establishes baseline conditions with all buses maintaining voltages within 0.95--1.05 p.u. boundaries. Critical buses 8, 9, and 23 demonstrate stable baseline operation at approximately 1.01, 1.02, and 1.03 p.u., respectively, providing optimal generation schedules $\mathbf{x}^*$ for adversarial assessment.

Stage 2 AAA model from \eqref{mod:AAA} identifies coordinated attacks targeting critical generators. The adversary optimization $J_2(\mathbf{x}^*, \mathbf{y})$ creates voltage degradation patterns where Bus 8 experiences the most severe depression to 0.97 p.u., while Buses 9 and 23 show reductions to 0.99 p.u. and 1.00 p.u., respectively. This coordinated targeting strategy demonstrates attack scenarios beyond traditional \textit{N-1} analysis scope, as the adversary systematically compromises generation.

Stage 3 RCRL controller from Algorithm~\ref{alg:trilevel_constrained_td3} achieves voltage recovery through real-time BESS coordination. The learned policy restores Bus 8 to 1.01 p.u. and maintains Buses 9 and 23 at acceptable levels throughout the operational horizon. This recovery validates the primal-dual augmented Lagrangian framework from \eqref{eq:augmented_lagrangian_vectorized} with millisecond response times.
\begin{figure}[!htb]
\vspace{-0.4cm}
\centering
\includegraphics[width=0.85\columnwidth]{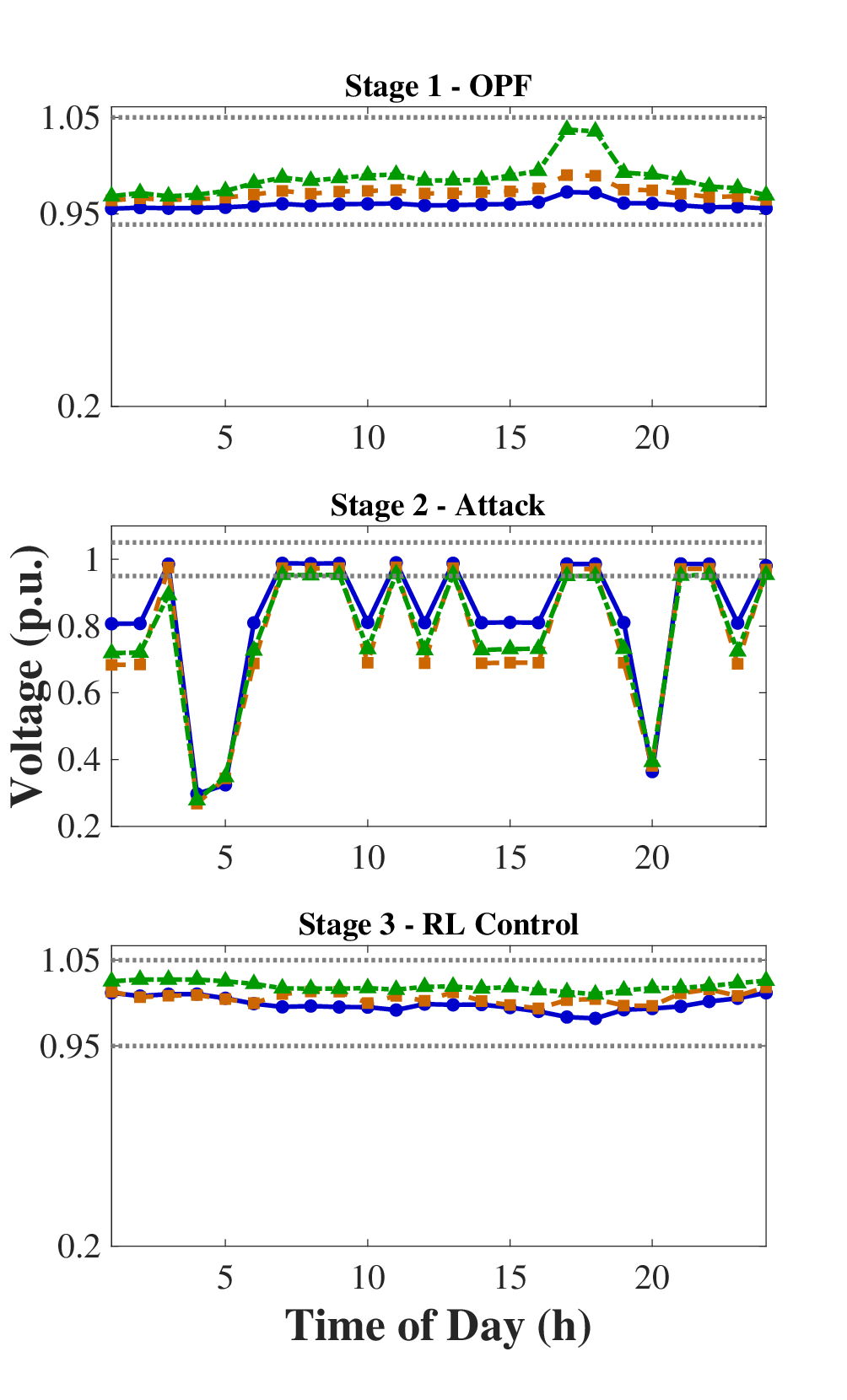}
\vspace{-0.4cm}
\caption{Voltage profiles across the stages for the IEEE 30-bus: 
{\color{blue}\textbf{---}} Bus 8,
{\color{orange}\textbf{$\square$---}} Bus 9,  
{\color{green}\textbf{$\triangle$---}} Bus 23,
{\color{gray}\textbf{$\boldsymbol{\cdots}$}} Limits (0.95/1.05 p.u.).}
\label{fig:trilevel_validation_compact}
\vspace{-0.4cm}
\end{figure}
\begin{figure}[!htb]
\centering
\includegraphics[width=0.85\columnwidth]{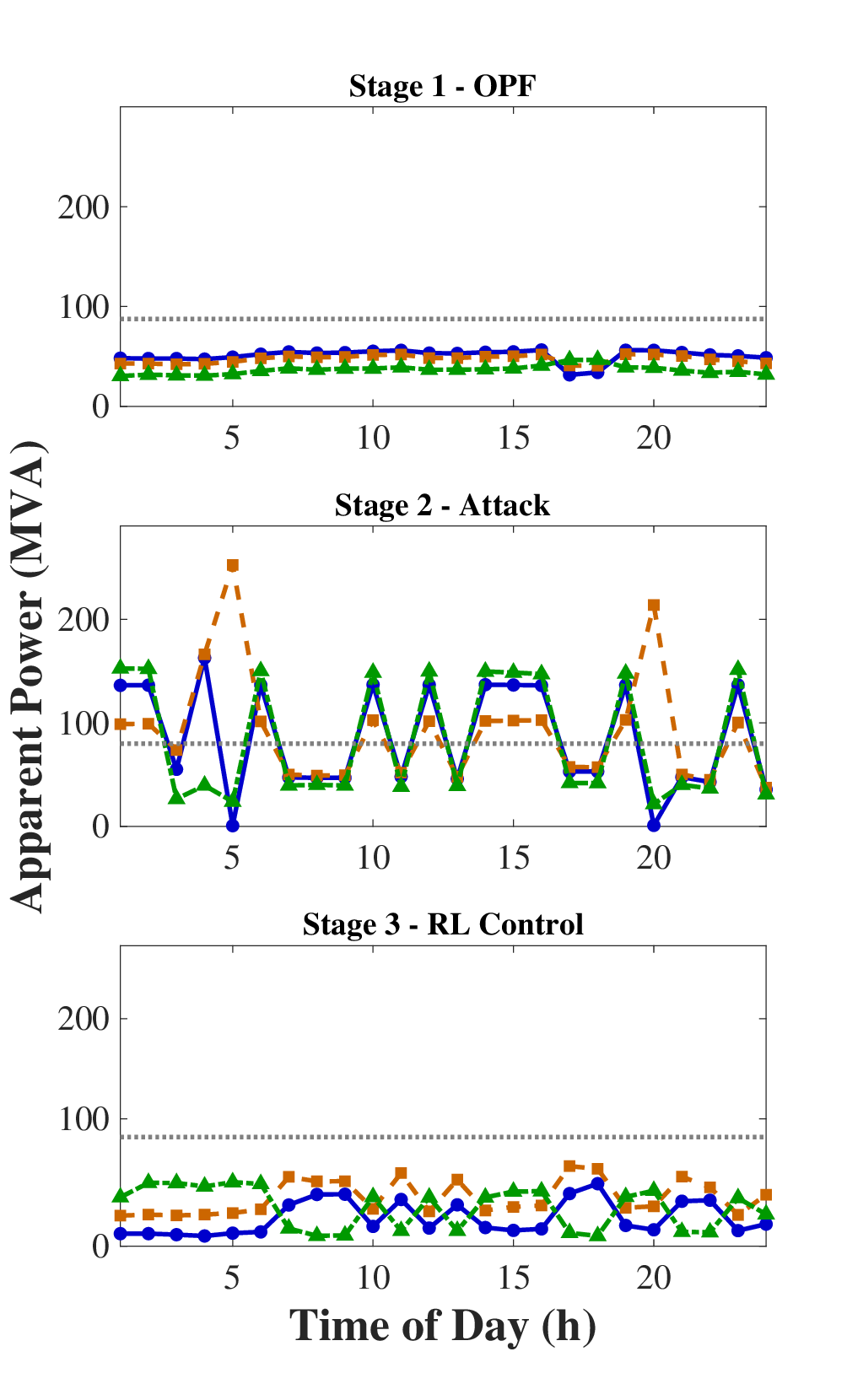}
\vspace{-0.4cm}
\caption{Power flow profiles across the stages for the IEEE 30-bus: 
{\color{blue}\textbf{---}} Line 5,
{\color{orange}\textbf{$\square$---}} Line 6,  
{\color{green}\textbf{$\triangle$---}} Line 14,
{\color{gray}\textbf{$\boldsymbol{\cdots}$}} Thermal limits.}
\label{fig:trilevel_validation_compact_line}
\vspace{-0.4cm}
\end{figure}

% \begin{figure}
% \centering
% \includegraphics[width=0.92\columnwidth]{Figures/Voltage_3stages_WS.eps}
% \vspace{-0.3cm}
% \caption{Voltage profiles across the stages for the IEEE 30-bus: 
% {\color{blue}\textbf{---}} Bus 8,
% {\color{orange}\textbf{$\square$---}} Bus 9,  
% {\color{green}\textbf{$\triangle$---}} Bus 23,
% {\color{gray}\textbf{$\boldsymbol{\cdots}$}} Limits (0.95/1.05 p.u.).}
% \label{fig:trilevel_validation_compact}
% \end{figure}

% \begin{figure}
% \centering
% \includegraphics[width=0.92\columnwidth]{Figures/Power_flow_3stages_WS.eps}
% \vspace{-0.3cm}
% \caption{Power flow profiles across the stages for the IEEE 30-bus: 
% {\color{blue}\textbf{---}} Line 5,
% {\color{orange}\textbf{$\square$---}} Line 6,  
% {\color{green}\textbf{$\triangle$---}} Line 14,
% {\color{gray}\textbf{$\boldsymbol{\cdots}$}} Thermal limits.}
% \label{fig:trilevel_validation_compact_line}

% \end{figure}

\subsubsection{Power Flow Response Analysis (Fig. \ref{fig:trilevel_validation_compact_line})}
Baseline operation (Stage 1) maintains conservative transmission loading with lines operating below thermal limits. The selected lines 5, 6, and 14 demonstrate loading patterns of approximately 8$\sim$35 MW, operating well in the secure margins.

The coordinated attack creates thermal stress across transmission lines through strategic generator targeting. Line 6 experiences the most significant loading increase, reaching approximately 252 MW during peak attack intensity, while Lines 5 and 14 show moderate increases to 162 and 152 MW, respectively. 

% The attack exploits power flow redistributions following generator outages, concentrating stress on critical transmission paths.

Real-time BESS coordination through the trained TD3 policy redistributes power flows within thermal constraints. The RCRL framework effectively manages Line 6 loading, reducing peak flows to acceptable levels while maintaining Lines 5 and 14 within operational bounds. This validates the Beta-Blending mechanism from \eqref{eq:beta_blending}, demonstrating successful transformation of computationally-intensive Stage 3 optimization into real-time simultaneous voltage and thermal management. The experimental validation confirms the elimination of both computational bottlenecks in existing security paradigms.

\section{Conclusion} \label{sec:conclusion}
This paper presents a tri-level robust constrained reinforcement learning framework for real-time defense against coordinated cyber-physical attacks on power systems. The framework addresses critical computational delays in identifying worst-case scenarios and coordinating defensive responses during which cascading failures propagate.
The approach integrates economic dispatch optimization, adversarial attack assessment, and constrained policy learning through three key innovations: (i) tri-level optimization identifying worst-case \textit{N-K} attack scenarios beyond traditional \textit{N-1} analysis, (ii) Beta-Blending safe exploration ensuring smooth transition from projection-based safety to constraint-aware deployment, and (iii) primal-dual augmented Lagrangian optimization embedding power system physics directly into policy learning for real-time constraint satisfaction without online optimization.

%=========== Original ===============================

% \section{Conclusion} \label{sec:conclusion}
% This paper presents a tri-level robust constrained reinforcement learning framework for real-time defense against coordinated cyber-physical attacks on power systems. The framework addresses the critical two-stage bottleneck in existing security paradigms: computational delays in identifying worst-case scenarios and optimization delays for defensive coordination during which cascading failures propagate.
% The framework integrates economic dispatch optimization, adversarial attack assessment, and constrained policy learning through three key innovations: (i) tri-level optimization that identifies worst-case N-K attack scenarios beyond traditional \textit{N-1} analysis, systematically generating diverse attack conditions affecting over 90\% of network buses, (ii) Beta-Blending safe exploration that ensures smooth transition from projection-based safety during training to constraint-aware deployment, eliminating computational overhead while maintaining feasibility, and (iii) primal-dual augmented Lagrangian optimization that embeds power system physics directly into policy learning, enabling real-time constraint satisfaction without online optimization.

% \vspace{-0.2cm}
\begin{footnotesize}
\bibliographystyle{IEEEtran}
\bibliography{all.bib}
\end{footnotesize}

\appendix
\section{Mathematical Proofs}
\label{appendix:proofs}

\renewcommand\qedsymbol{$\blacksquare$}

\begin{proof}[Proof of Theorem~\ref{thm:primal_dual_convergence}]
We analyze the primal-dual sequence $\{(\phi^{(k)}, \boldsymbol{\lambda}_t^{(k)}, \boldsymbol{\mu}_t^{(k)})\}$ generated by the update rules \eqref{eq:primal_update}--\eqref{eq:practical_dual_updates} to demonstrate that all accumulation points satisfy the first-order stationarity conditions. Unlike [8], which considers single-stage SDOPF constraints, our proof incorporates the Beta-blending projection mechanism and tri-level constraint dependencies.

\textbf{Step 1: Accumulation Point Existence.}
The policy parameter space $\Phi$ is compact under assumption (iii), while the dual variables maintain boundedness through projection operations under assumption (iv). This yields the bound relationships in \eqref{eq:dual_bounds}.

\begin{equation}
\label{eq:dual_bounds}
\|\boldsymbol{\lambda}_t^{(k)}\|_{\infty} \leq \Lambda_{\max}, \quad \|\boldsymbol{\mu}_t^{(k)}\|_{\infty} \leq M_{\max}
\end{equation}

Consequently, the generated sequence remains within the compact set specified in \eqref{eq:compact_set}

\begin{equation}
\label{eq:compact_set}
\mathcal{K} = \Phi \times [-\Lambda_{\max}, \Lambda_{\max}]^{m_e} \times [0, M_{\max}]^{m_i}
\end{equation}

The Bolzano-Weierstrass theorem guarantees that any bounded sequence in finite-dimensional space contains at least one accumulation point $(\phi^*, \boldsymbol{\lambda}_t^*, \boldsymbol{\mu}_t^*)$.

\textbf{Step 2: Gradient Lipschitz Continuity.}
Assumption (i) ensures the existence of a Lipschitz constant $L > 0$ such that for arbitrary $\phi_1, \phi_2 \in \Phi$ and fixed dual variables within $\mathcal{K}$, the gradient mapping satisfies the Lipschitz continuity condition in \eqref{eq:lipschitz_condition}.

\begin{equation}
\label{eq:lipschitz_condition}
\|\nabla_{\phi} \mathcal{L}_{\mathcal{A}}(\phi_1, \boldsymbol{\lambda}, \boldsymbol{\mu}) - \nabla_{\phi} \mathcal{L}_{\mathcal{A}}(\phi_2, \boldsymbol{\lambda}, \boldsymbol{\mu})\| \leq L\|\phi_1 - \phi_2\|
\end{equation}

This property emerges from the twice continuously differentiable policy function combined with the boundedness of dual variables.

\textbf{Step 3: Beta-blending Impact Analysis.}
The Beta-blending mechanism introduces time-varying projection effects through $\beta_t = \min(t/T_{\beta}, 1)$. For $t < T_{\beta}$, the final action $\mathbf{a}_t^{\text{final}} = \beta_t \mathbf{a}_t^{\text{expl}} + (1-\beta_t)\mathcal{P}_{\mathcal{F}_t}(\mathbf{a}_t^{\text{expl}})$ combines exploratory and projected actions. As $t \geq T_{\beta}$, we have $\beta_t = 1$, eliminating projection effects and recovering the standard gradient flow. This ensures that asymptotic convergence properties remain unaffected by the initial projection operations.

\textbf{Step 4: Monotonic Descent Property.}
The primal parameter update rule specified in \eqref{eq:primal_update}, when combined with the step size restriction $\eta_{\phi} \leq \frac{1}{2L}$, guarantees the monotonic descent property expressed in \eqref{eq:descent_property}.

\begin{align}
\label{eq:descent_property}
\mathcal{L}_{\mathcal{A}}(\phi^{(k+1)}, \boldsymbol{\lambda}_t^{(k)}, \boldsymbol{\mu}_t^{(k)}) &\leq \mathcal{L}_{\mathcal{A}}(\phi^{(k)}, \boldsymbol{\lambda}_t^{(k)}, \boldsymbol{\mu}_t^{(k)}) \nonumber \\
&\quad - \frac{\eta_{\phi}}{2}\|\nabla_{\phi} \mathcal{L}_{\mathcal{A}}(\phi^{(k)}, \boldsymbol{\lambda}_t^{(k)}, \boldsymbol{\mu}_t^{(k)})\|^2
\end{align}

\textbf{Step 5: Gradient Norm Summability.}
Given that the augmented Lagrangian function is bounded from below on the compact set $\mathcal{K}$ and the descent property is maintained, the sequence of gradient norms satisfies the summability criterion in \eqref{eq:summability}.

\begin{equation}
\label{eq:summability}
\sum_{k=0}^{\infty} \|\nabla_{\phi} \mathcal{L}_{\mathcal{A}}(\phi^{(k)}, \boldsymbol{\lambda}_t^{(k)}, \boldsymbol{\mu}_t^{(k)})\|^2 < \infty
\end{equation}

This summability condition directly implies the gradient convergence property shown in \eqref{eq:gradient_convergence}.

\begin{equation}
\label{eq:gradient_convergence}
\lim_{k \to \infty} \|\nabla_{\phi} \mathcal{L}_{\mathcal{A}}(\phi^{(k)}, \boldsymbol{\lambda}_t^{(k)}, \boldsymbol{\mu}_t^{(k)})\| = 0
\end{equation}

\textbf{Step 6: Stationarity Conditions at Accumulation Points.}
Consider any convergent subsequence $\{(\phi^{(k_j)}, \boldsymbol{\lambda}_t^{(k_j)}, \boldsymbol{\mu}_t^{(k_j)})\}$ converging to the limit $(\phi^*, \boldsymbol{\lambda}_t^*, \boldsymbol{\mu}_t^*)$. Through the continuity of the gradient mapping and the convergence result \eqref{eq:gradient_convergence}, every accumulation point must satisfy the stationarity condition presented in \eqref{eq:final_stationarity}.

\begin{equation}
\label{eq:final_stationarity}
\nabla_{\phi} \mathcal{L}_{\mathcal{A}}(\phi^*, \boldsymbol{\lambda}_t^*, \boldsymbol{\mu}_t^*) \in \mathcal{N}_{\Phi}(\phi^*)
\end{equation}

The key distinction from [8] is that our augmented Lagrangian incorporates cascaded constraints where Stage 3 feasibility depends on attack scenarios from Stages 1-2, requiring modified constraint qualification conditions.

\end{proof}

\begin{proof}[Proof of Corollary~\ref{cor:beta_blending_convergence}]
For the Beta-blending schedule $\beta_t = \min(t/T_{\beta}, 1)$, we analyze the impact on convergence properties in two phases:

\textbf{Phase 1} ($t < T_{\beta}$): The blended action contains projection components that may perturb the gradient flow. However, since the projection operator maps to the feasible set, all intermediate iterates remain within the constraint-admissible region.

\textbf{Phase 2} ($t \geq T_{\beta}$): As $\beta_t = 1$, we have $\mathbf{a}_t^{\text{final}} = \mathbf{a}_t^{\text{expl}}$, eliminating projection effects. The gradient flow reduces to the standard augmented Lagrangian case, preserving the stationarity conditions from Theorem~\ref{thm:primal_dual_convergence}.

The continuous transition ensures no discontinuous jumps in the policy gradient, maintaining convergence properties.
\end{proof}

\begin{proof}[Proof of Theorem~\ref{thm:constraint_satisfaction}]
Building upon the conditions established in Theorem~\ref{thm:primal_dual_convergence} and incorporating Slater's constraint qualification, we derive explicit upper bounds on constraint violations at stationary points. This analysis extends [8] by considering the cascaded constraint structure of our tri-level optimization framework.

\textbf{Step 1: Stationarity Condition Analysis.}
At any stationary point $(\phi^*, \boldsymbol{\lambda}_t^*, \boldsymbol{\mu}_t^*)$, the first-order optimality condition can be decomposed into the form presented in \eqref{eq:stationarity_decomposition}.

\begin{align}
\label{eq:stationarity_decomposition}
\nabla_{\phi} Q_{\psi_1}(\mathbf{s}_t, \pi_{\phi^*}(\mathbf{s}_t)) &= \sum_{j=1}^{m_e}(\lambda_{t,j}^* + \rho_{h,j} h_j^*)\nabla_{\phi} h_j^* \nonumber \\
&\quad + \sum_{k=1}^{m_i}(\mu_{t,k}^* + \rho_{g,k} [g_k^*]_+)\nabla_{\phi} [g_k^*]_+
\end{align}

where $h_j^* := h_j(\mathbf{s}_t, \pi_{\phi^*}(\mathbf{s}_t))$ and $g_k^* := g_k(\mathbf{s}_t, \pi_{\phi^*}(\mathbf{s}_t))$ represent the constraint function values at the stationary point.

\textbf{Step 2: Constraint Qualification Application.}
Under the Linear Independence Constraint Qualification specified in assumption (ii), the constraint gradients exhibit linear independence, which ensures the existence of positive constants $\underline{c}, \overline{c} > 0$ satisfying the bounds shown in \eqref{eq:licq_bounds}.

\begin{subequations}
\label{eq:licq_bounds}
\begin{align}
\underline{c} \leq \|\nabla_{\phi} h_j^*\| \leq \overline{c}, &\quad \forall j \in \{1, \ldots, m_e\} \label{eq:licq_equality} \\
\underline{c} \leq \|\nabla_{\phi} [g_k^*]_+\| \leq \overline{c}, &\quad \forall k \in \mathcal{A}(\phi^*) \label{eq:licq_inequality}
\end{align}
\end{subequations}

where $\mathcal{A}(\phi^*) := \{k : g_k^* > 0\}$ identifies the active constraint index set.

\textbf{Step 3: Penalty Parameter Lower Bound Derivation.}
For any equality constraint $j$ satisfying $|h_j^*| > \epsilon$, we apply the inner product operation between equation \eqref{eq:stationarity_decomposition} and $\nabla_{\phi} h_j^*$, followed by the Cauchy-Schwarz inequality, resulting in the bound relationship in \eqref{eq:penalty_derivation}.

\begin{align}
\label{eq:penalty_derivation}
|\rho_{h,j} h_j^*| \underline{c}^2 &\leq \|\nabla_{\phi} Q_{\psi_1}\| \cdot \overline{c} + \Lambda_{\max} \overline{c}^2 + C_{\text{cross}} \nonumber \\
\Rightarrow |h_j^*| &\leq \frac{\|\nabla_{\phi} Q_{\psi_1}\| \cdot \overline{c} + \Lambda_{\max} \overline{c}^2 + C_{\text{cross}}}{\rho_{h,j} \underline{c}^2}
\end{align}

To ensure the constraint violation $|h_j^*| \leq \epsilon$, the penalty parameter must exceed the threshold specified in \eqref{eq:penalty_lower_bound}.

\begin{equation}
\label{eq:penalty_lower_bound}
\rho_{h,j} \geq \frac{\|\nabla_{\phi} Q_{\psi_1}\| + \Lambda_{\max}}{\epsilon}
\end{equation}

\textbf{Step 4: Inequality Constraint Extension.}
For active inequality constraints $k \in \mathcal{A}(\phi^*)$ with violations $g_k^* > \epsilon$, an analogous analytical approach produces the penalty parameter lower bound in \eqref{eq:inequality_penalty_bound}.

\begin{equation}
\label{eq:inequality_penalty_bound}
\rho_{g,k} \geq \frac{M_{\max}}{\epsilon}
\end{equation}

\textbf{Step 5: Global Constraint Violation Bounds.}
By computing the Euclidean norms across all constraint functions and applying penalty parameters that exceed the thresholds defined in \eqref{eq:penalty_thresholds}, the aggregate constraint violations satisfy the bounds specified in \eqref{eq:violation_bounds}.

\begin{equation}
\label{eq:violation_bounds}
\|\mathbf{r}_{\lambda,t}\|^2 = \sum_{j=1}^{m_e}(h_j^*)^2 \leq m_e \epsilon^2, \quad \|\mathbf{r}_{\mu,t}\|^2 = \sum_{k=1}^{m_i}([g_k^*]_+)^2 \leq m_i \epsilon^2
\end{equation}

Consequently, the selection of penalty parameters according to the guidelines in \eqref{eq:penalty_thresholds} establishes the desired constraint violation bounds presented in \eqref{eq:constraint_bounds}.
\end{proof}

\section{Implementation Considerations}
\label{appendix:remarks}

\begin{remark}
 The derived penalty parameter bounds offer systematic guidance for hyperparameter selection in practical implementations. The recommended approach begins with the minimum threshold values and employs adaptive increases when constraint violations exceed predetermined tolerance levels.   
\end{remark}

\begin{remark}
  The Linear Independence Constraint Qualification assumption generally holds in power system applications during normal operating conditions, since the constraint gradients correspond to physically meaningful power flow sensitivities that naturally exhibit linear independence properties.  
\end{remark}

\begin{remark}
    The presented convergence analysis establishes local convergence properties. Achieving global convergence guarantees would necessitate additional convexity assumptions that are fundamentally incompatible with the inherently nonlinear characteristics of AC power flow constraints.
\end{remark}

\begin{remark}
  While the projection operations specified in \eqref{eq:practical_dual_updates} ensure algorithmic stability, they may introduce computational overhead in practical implementations. Alternative approaches, such as approximate projection methods or barrier function techniques, can provide improved computational efficiency.  
\end{remark}

\begin{remark}
   The step size restriction $\eta_{\phi} \leq \frac{1}{2L}$ in \eqref{eq:primal_update} requires accurate estimation of the Lipschitz constant $L$. When this constant is unknown or difficult to estimate, practical implementations can employ adaptive step size algorithms or line search methodologies as viable alternatives. 
\end{remark}

\begin{remark}
  The relationship between our theoretical results and [8] lies in the constraint structure complexity. While [8] addresses standard SDOPF constraints that are independent across time steps, our tri-level framework handles interdependent constraints where Stage 3 feasibility depends on attack scenarios identified in Stages 1-2. The Beta-blending mechanism enables this increased complexity without sacrificing convergence guarantees.   
\end{remark}

%========================

\end{document}